\documentclass[a4paper,onecolumn,accepted=2025-07-11]{quantumarticle}
\pdfoutput=1

\usepackage[utf8]{inputenc}
\usepackage[english]{babel}

\usepackage{xcolor}
\definecolor{mygreen}{RGB}{0, 111, 98}
\definecolor{myorange}{RGB}{255, 117, 0}

\usepackage[hidelinks]{hyperref}       
\hypersetup{colorlinks=true, urlcolor=mygreen, linkcolor=myorange, citecolor=mygreen }
\usepackage{url}

\usepackage{booktabs}
\usepackage{array}
\usepackage{makecell}
\usepackage{multirow}

\usepackage{amsmath}
\usepackage{amsthm}
\usepackage{amsfonts}
\usepackage{amssymb}
\usepackage{braket}
\usepackage{bm}

\usepackage{graphicx}
\usepackage{caption}
\usepackage{float}

\usepackage{tikz}
\usetikzlibrary{shadows}
\usetikzlibrary{matrix, positioning}

\usepackage[numbers,compress]{natbib}
\usepackage[compat=0.6]{yquant}
\useyquantlanguage{groups}
\usetikzlibrary{fit,quotes,positioning}
\pgfdeclareshape{yquant-multictrlshape}{
	\inheritsavedanchors[from=yquant-circle]
	\foreach \anc in {center, north, north east, east, south east, south, south west, west, north west} {
		\inheritanchor[from=yquant-circle]{\anc}
	}
	\inheritanchorborder[from=yquant-circle]
	\backgroundpath{
		\pgfpathmoveto{\pgfqpoint{-0.707107\dimexpr\xradius\relax}{-0.707107\dimexpr\yradius\relax}}
		\pgfpathlineto{\pgfqpoint{0.707107\dimexpr\xradius\relax}{0.707107\dimexpr\yradius\relax}}
		\pgfpathellipse{\pgfpointorigin}
		{\pgfqpoint{\xradius}{0pt}}
		{\pgfqpoint{0pt}{\yradius}}
	}
	\inheritclippath[from=yquant-circle]
}
\yquantset{multictrl/.style={every control/.style={shape=yquant-multictrlshape,radius=1.05mm}}}
\yquantset{plusctrl/.style={/yquant/every control/.style={/yquant/operators/every not}, every positive control/.style={}}}

\usepackage{qcircuit}
\usepackage{environ}
\newcommand{\myqctmp}[2][0.25]{\Qcircuit @C=#2em @R=#1em @!R}
\NewEnviron{myqcircuit}[1][0.25]{\vcenter{\myqctmp[#1]{0.75} {\BODY}}}
\NewEnviron{myqcircuit*}[2]{\vcenter{\myqctmp[#1]{#2} {\BODY}}}

\newtheorem{theorem}{Theorem}[section]
\newtheorem{definition}[theorem]{Definition}
\newtheorem{lemma}[theorem]{Lemma}
\newtheorem{corollary}[theorem]{Corollary}
\newtheorem{remark}{Remark}
\usepackage{enumitem}
\newenvironment{assumption}[1][]
  {\begin{enumerate}[label=(A\arabic*), ref=A\arabic*, #1]}
  {\end{enumerate}}

\numberwithin{equation}{section}

\usepackage{ulem}

\author{Chunlin Yang $^*$}
\affiliation{School of Mathematical and Sciences, Harbin Engineering University, China}
\author{Zexian Li $^*$}
\affiliation{Department of Applied Mathematics, The Hong Kong Polytechnic University, China}
\author{Hongmei Yao}
\affiliation{School of Mathematical and Sciences, Harbin Engineering University, China}
\email{hongmeiyao@163.com}
\author{Zhaobing Fan}
\affiliation{School of Mathematical and Sciences, Harbin Engineering University, China}
\author{Guofeng Zhang}
\affiliation{Department of Applied Mathematics, The Hong Kong Polytechnic University, China}
\author{Jianshe Liu}
\affiliation{College of Underwater Acoustic Engineering, Harbin Engineering University, China}

\begin{document}
	
\title{Dictionary-based Block Encoding of Sparse Matrices with Low Subnormalization and Circuit Depth}

\maketitle
\def\thefootnote{*}\footnotetext{These authors contributed equally to this work}\def\thefootnote{\arabic{footnote}}
\begin{abstract}
    Block encoding severs as an important data input model in quantum algorithms, enabling quantum computers to simulate non-unitary operators effectively. In this paper, we propose an efficient block-encoding protocol for sparse matrices based on a novel data structure, called the dictionary data structure, which classifies all non-zero elements according to their values and indices. Non-zero elements with the same values, lacking common column and row indices, belong to the same classification in our block-encoding protocol's dictionary. When compiled into the \{\rm U(2), CNOT\} gate set, the protocol queries a $2^n \times 2^n$ sparse matrix with $s$ non-zero elements at a circuit depth of $\mathcal{O}(\log(ns))$, utilizing $\mathcal{O}(n^2s)$ ancillary qubits. This offers an exponential improvement in circuit depth relative to the number of system qubits, compared to existing methods~\cite{clader2022quantum,zhang2024circuit} with a circuit depth of $\mathcal{O}(n)$. Moreover, in our protocol, the subnormalization, a scaled factor that influences the measurement probability of ancillary qubits, is minimized to $\sum_{l=0}^{s_0}\vert A_l\vert$, where $s_0$ denotes the number of classifications in the dictionary and $A_l$ represents the value of the $l$-th classification. Furthermore, we show that our protocol connects to linear combinations of unitaries (LCU) and the sparse access input model (SAIM). To demonstrate the practical utility of our approach, we provide several applications, including Laplacian matrices in graph problems and discrete differential operators.
\end{abstract}

\tableofcontents
	
\section{Introduction}
\label{section: introduction}
Quantum algorithms have demonstrated unprecedented potential in solving classically intractable problems, exemplified by landmark algorithms such as the Deutsch-Jozsa algorithm~\cite{deutsch1992rapid}, Shor's algorithm~\cite{shor1994algorithms}, and the HHL algorithm~\cite{harrow2009quantum}. A critical enabler of these advancements lies in efficient data input models, which bridge classical information with quantum processing. Commonly used models include block encoding~\cite{gilyen2019quantum}, SAIM~\cite{harrow2009quantum, gilyen2019quantum, berry2007efficient, childs2011simulating, childs2010on, berry2014exponential, childs2017quantum, chakraborty2019thepower, babbush2023exponential}, and LCU~\cite{LongGui-Lu_2006,childs2012hamiltonian}. Recent studies have shown that many quantum algorithms can be unified and reconstructed within the framework of quantum singular value transformation (QSVT) \cite{gilyen2019quantum, martyn2021grand}. These algorithms rely on the block encoding, enabling the embedding of arbitrary matrices into unitary operators,
\begin{equation*}
	U_{A} = \begin{pmatrix}
		A/\alpha & * \\
		* & *
	\end{pmatrix},
\end{equation*}
where $*$ denotes a matrix block yet to be determined. The factor $\alpha$, called \textit{subnormalization}, scales $A$ to ensure $\left\|A/\alpha\right\|_{\rm sp} \leq 1$, since the singular values of any matrix block of a unitary must not be larger than $1$, where $\left\|\cdot\right\|_{\rm sp}$ denotes the spectral norm.

Given the important role of block encoding in quantum computation, substantial research efforts have been dedicated to optimizing its implementation. Special matrices including density operators, the hierarchical matrices, pseudo-differential operators, and pairing Hamiltonians, have been studied in \cite{low2019hamiltonian,van2019improvements,nguyen2022block,li2023efficient,liu2025efficient}.
Generally, for dense matrices, Kerenidis and Prakash~\cite{kerenidis2020quantum}, and Chakraborty et al.~\cite{chakraborty2019thepower} showed how to implement block encoding efficiently of matrices stored in quantum data structures. Based on the quantum random access memory query model, Clader et al.~\cite{clader2022quantum} developed several block-encoding methods for a dense matrix of classical data. Using single- and two-qubit gates, a method, called FABLE~\cite{camps2022fable,parker2024sfable}, was proposed to generate approximate quantum circuits for block-encoding matrices in a fast manner, while Li et al.~\cite{li2025binary} improved it by fast computational methods for block encoding classical matrices with few ancillary qubits, low computing time, and low quantum gate counts. For sparse matrices, Gily{\'e}n et al.~\cite{gilyen2019quantum} introduced a foundational block encoding, which was subsequently improved by using preamplification. If a sparse matrix has a well-defined structure, Camps et al.~\cite{camps2024explicit} presented a scheme to efficiently block encode it. Furthermore, if a sparse matrix has an arithmetical structure, S\"{u}nderhauf et al.~\cite{sunderhauf2024block} developed a novel block encoding scheme and provided two improved schemes using preamplification and state preparation to reduce \textit{subnormalization}.

The structural characteristics of matrices significantly influence the implementation and resource requirements of algorithms, which in turn affect their efficiency. Sparse structured matrices, in particular, stand out due to their substantial number of zero elements and the regular distribution of non-zero elements. They have widely used applications, including fluid mechanics \cite{aki2002quantitative, jensen2011computational}, image processing \cite{xiao2021bayesian, yuan2024simultaneous}, and network analysis \cite{abusalah2020accelerated, shen2025weak}. Thanks to the characteristics, it is possible to store only non-zero elements to compress storage. 

Despite the advantages of sparse structured matrices, existing block encoding schemes face certain limitations when dealing with them. Block encoding in \cite{camps2024explicit} requires that every data value should appear in all columns, while the PREP/UNPREP scheme in \cite{sunderhauf2024block} requires $\lceil \log D\rceil\leq \lceil\log S_{c}\rceil,\lceil\log S_{r}\rceil$, where $D$ is the number of distinct data values, $S_{c},S_{r}$ are the maximum column and row sparsities, respectively. These constraints limit the applicability and flexibility of block encoding schemes. To address these issues, this paper proposes a block encoding scheme of sparse matrices based on a dictionary data structure. Crucially, the block-encoding protocol optimizes the trade-off between \textit{subnormalization} and \textit{circuit depth}. This enhancement allows for more efficient and flexible handling of sparse structured matrices in quantum algorithms, advancing the practical application of quantum computing.

The complexity of a block encoding scheme is related to the design of the scheme and the properties of the encoded matrix. To assess the cost of a block encoding, a \textit{figure of merit} \cite{sunderhauf2024block} is defined as
\begin{equation*}\label{equation: figure of merit}
	(T\mbox{\textit{-gate count}}) \cdot \mbox{\textit{subnormalization}}.
\end{equation*}
The \textit{subnormalization} is a factor to scale the matrix. A lower \textit{subnormalization} can increase the probability to measure $\ket{0}$ for ancillary qubits and lead to shorter circuits of algorithms within the framework of block encoding. The \textit{circuit depth} can reflect the time consumption of implementing quantum circuits. In this article, inspired by above \textit{figure of merit}, we introduce the definition of \textit{time metric} for block-encoding protocols, 
\begin{equation}\label{time metric}
	\mbox{\textit{time metric}} = \mbox{\textit{circuit depth}} \times \mbox{\textit{subnormalization}},
\end{equation}
as the quantitative indicator of the time complexity, where the \textit{circuit depth} is the maximum number of layers of element gates in a quantum circuit. A lower \textit{circuit depth} can enhance the feasibility and efficiency of quantum algorithms. In algorithms within the framework of QSVT, the time complexity is commonly characterized by the query complexity of block encoding, which can be directly translated into \textit{circuit depth}. This correspondence establishes that the \textit{time metric} defined in Equation~\eqref{time metric} provides a unified framework to analytically quantify the time complexity of quantum algorithms within the framework of QSVT.
 
Our main contribution in this paper is summarized as follows. 
\begin{itemize}  
    \item \textbf{Dictionary Data Structure:} We introduce a new data structure for sparse structured matrices, namely, dictionary data structure, which is shown in Table~\ref{dictionary}. It depends on the classification of all triplets of non-zero matrix elements. Within this framework, we can unify several existing block encoding schemes of sparse matrices, details in Appendix~\ref{subsection: several sparse block encodings within dictionary}.

    \item \textbf{Dictionary-based block encoding of sparse matrices:} Based on the dictionary data structure in Table \ref{dictionary}, we provide a dictionary-based block encoding of sparse matrices in Theorem~\ref{theorem: sparse block encoding} with the \textit{subnormalization} $\sum_{l=0}^{s_0-1}\left|A_l\right|$, where $s_0$ is the number of data items. For non-negative symmetric matrices, we extend the Hermitian version in Theorem~\ref{theorem: sparse Hermitian block encoding} based on the dictionary data structure in Table~\ref{table: symmetric dictionary data structure in this paper}. We also demonstrate that our dictionary-based block encoding of sparse matrices is a linear combination of unitaries. 

    \item \textbf{Circuit Depth Optimization:} 
    We analyze the \textit{circuit depth} of the dictionary-based block encoding of sparse matrices, achieving a scaling of $\mathcal{O}\big(\log(ns)\big)$ (Theorem~\ref{theorem: circuit depth}), where $n$ denotes the number of qubits encoding the matrix, $s$ is the number of non-zero elements in the matrix. We also analyze the \textit{circuit depth} of the PREP/UNPREP block encoding in \cite{sunderhauf2024block} as well as the block encoding using controlled state preparation in \cite{clader2022quantum}. These results are summarized in Table~\ref{table: time metric}.
\end{itemize}

\begin{table}[htbp]
    \centering
    \begin{tabular}{!{\vrule width \heavyrulewidth}c!{\vrule width \heavyrulewidth}c!{\vrule width \heavyrulewidth}c!{\vrule width \heavyrulewidth}c!{\vrule width \heavyrulewidth}c!{\vrule width \heavyrulewidth}}
        \hline\cline{1-5}
        \textbf{Block Encoding} & \multicolumn{2}{c!{\vrule width \heavyrulewidth}}{\textbf{Circuit Depth}} & \textbf{Ancillary qubits} & \textbf{Subnormalization} \\
        \hline\cline{1-5}
        \makecell{PREP/UNPREP~\cite{sunderhauf2024block}} & Lemma \ref{lemma: circuit depth of Sunderhauf} & $\mathcal{O}\left(n2^{n/2}\right)$ & $\Omega(4^n/n)$ & $\frac{\sqrt{S_{c}S_{r}}}{D}\sum_{d=0}^{D-1}\left|A_{d}\right|$  \\
        \hline
        \makecell{~\cite{clader2022quantum,zhang2024circuit}} & Lemma~\ref{lemma: comparison depth 2} & $\mathcal{O}\left(n\right)$ & $\mathcal{O}(2^{2n})$ & $\Vert A\Vert_F$   \\
        \hline
        \makecell{Dictionary-based \\ block encoding \\ (\textbf{This work)}} & Theorem \ref{theorem: circuit depth} & \textbf{$\mathcal{O}\left(\log(ns)\right)$} & $\mathcal{O}(n^2s)$ & \textbf{$\sum_{l=0}^{s_{0}-1}\left|A_l\right|$} \\
        \hline\cline{1-5}
    \end{tabular}
    \caption{\textit{Circuit depth} and \textit{subnormalization} of block-encoding protocols for encoding a sparse structured matrix $A \in\mathbb{C}^{2^n \times 2^n}$ with $s$ non-zero elements. The matrix $A$ can be encoded by a dictionary data structure specified in Table~\ref{table: dictionary data structure in this paper} with $s_0$ data items under our block encoding (or $D$ data items under PREP/UNPREP block encoding~\cite{sunderhauf2024block}). Here, $S_c$ and $S_r$ denote column and row sparsity, respectively. The number of data items is always not large than the number of non-zero elements in a matrix, that is, $s_0 \leq s$. For matrices with repeated elements, our approach achieves superior \textit{subnormalization} compared to the low-depth block-encoding protocols proposed in~\cite{clader2022quantum,zhang2024circuit}. This advantage is particularly evident in the \textit{time metric} defined in Equation~\eqref{time metric}, where our protocols offer an exponential improvement.}
    \label{table: time metric}
\end{table}

This paper is organized as follows. Section \ref{section: notations and conventions} establishes foundational notations and conventions. In Section \ref{section: block encoding}, we develop a dictionary data structure framework and propose a dictionary-based block encoding of sparse matrices with reduced \textit{subnormalization}, subsequently extending it to the Hermitian version for non-negative symmetric matrices. We further establish connections between this block encoding scheme and LCU, while rigorously analyzing its logarithmic \textit{circuit depth} implementation. Comparative benchmarks are provided against the PREP/UNPREP protocol in \cite{sunderhauf2024block} and controlled state preparation methods in \cite{clader2022quantum}. Section \ref{section: applications} demonstrates practical implementations for the signless Laplacian matrices in graphs and matrices of discrete differential operators, verifying the feasibility of the proposed block encoding of sparse matrices with Python code (\href{https://github.com/ChunlinYangHEU/DIBLE}{https://github.com/ChunlinYangHEU/DIBLE}). Conclusions and future directions are presented in Section \ref{section: conclusion and outlook}. Appendix~\ref{section: details of block encoding} details existing dictionary-based block-encoding protocols of sparse matrices, with \textit{circuit depth} comparisons in Appendix~\ref{section: Circuit Depth Comparison}. The application to generalized eigenvalue problems (GEPs) in ocean acoustics appears in Appendix~\ref{section: GEPs in ocean acoustics}.

\section{Notations and Conventions}
\label{section: notations and conventions}
In this section, we introduce the notations and conventions used throughout this paper. 

Let $\left[m,n\right] = \left\{m,m+1,\cdots,n\right\}$, with the convention that $\left[n\right] =\left[1,n\right]$. $\ket{0}$ denotes the vector $e_{0} = \left(1,0\right)^{T}$, and $\ket{1}$ denotes the vector $e_{1} = \left(0,1\right)^{T}$. The $m$-fold tensor product $\ket{0}\otimes\cdots\otimes\ket{0}$ is compactly denoted $\ket{0}^{\otimes m}$. We denote the $N \times N$ identity matrix by $I_{N}$. For the special case of $2 \times 2$ identity matrices (common in quantum computing), we simply write $I \equiv I_2$. The row and column indices of a matrix always start from $0$ in this paper. The $j$-th column of $I_{N}$ is denoted by $\left|j\right>$, where $j\in \left[0,N-1\right]$. For a nonnegative integer $j\in\left[0,2^{n}-1\right]$, it has a binary representation
\begin{equation*}
	j = \left[ j_{n-1} \cdots j_{1} j_{0} \right] = j_{n-1} \times 2^{n-1} + \cdots + j_{1} \times 2^{1} + j_{0} \times 2^{0},
\end{equation*}
where $j_{k}\in\left\{0,1\right\}$, $k\in \left[0,n-1\right]$.  For a $1$-bit binary number $a$, $\overline{a}$ denotes the NOT operation of $a$. We denote modulo 2 addition/subtraction as $\oplus$ using an exclusive OR (XOR) operation on the corresponding
binary digits of each operand. For the $n$-bit binary operation, it has
$$
{\rm XOR}(j,0) = j\oplus 0 = j,\quad {\rm XOR}(j,j) = j\oplus j=0,\quad \forall j\in\mathbb{Z}.
$$
The state of a qubit is a superposition of $\ket{0}$ and $\ket{1}$, and a nonnegative integer $j = \left[j_{n-1}\cdots j_{1}j_{0}\right]$ can be prepared as a set of quantum states $\ket{j}\equiv\ket{j_{n-1}} \otimes \cdots \ket{j_{1}} \otimes \ket{j_{0}}$. For a complex $c=\vert c\vert e^{\imath\theta}$, its square root is uniquely defined as $\sqrt{c} = \sqrt{\left|c\right|}e^{\imath\theta/2}$.

The letters $H$, $X$, $Y$, and $Z$ are used to represent the Hadamard, Pauli-$X$, Pauli-$Y$, and Pauli-$Z$ matrices, respectively, which are defined as follows,
\begin{equation*}
	H=\frac{1}{\sqrt{2}} \begin{pmatrix}1&1\\1&-1\end{pmatrix},\quad
	X= \begin{pmatrix}0&1\\1&0\end{pmatrix},\quad
	Y= \begin{pmatrix}0&-\imath \\\imath&0\end{pmatrix},\quad
	Z= \begin{pmatrix}1&0\\0&-1\end{pmatrix}.
\end{equation*}

\begin{table}[htbp]
\centering
    \begin{tabular}{!{\vrule width \heavyrulewidth}c|c|c!{\vrule width \heavyrulewidth}}
        \hline \cline{1-3}
        \textbf{Notation} & \textbf{Mathematical form} & \textbf{Circuit} \\ 
        \hline \cline{1-3}
        $0$-CNOT
        & $ \ket{0}\bra{0}\otimes X + \ket{1}\bra{1}\otimes I=
        \begin{pmatrix}
            0 & 1 & 0 & 0\\
            1 & 0 & 0 & 0\\
            0 & 0 & 1 & 0\\
            0 & 0 & 0 & 1\\
        \end{pmatrix}
        $
        &
        \makecell{\begin{tikzpicture}
            \begin{yquant}
                qubit {} q0;
                qubit {} q1;
                cnot q1 ~ q0;
            \end{yquant}
		\end{tikzpicture}} \\ 
        \hline
        $1$-CNOT
        &
        $\ket{1}\bra{1}\otimes X + \ket{0}\bra{0}\otimes I=
        \begin{pmatrix}
            1 & 0 & 0 & 0\\
            0 & 1 & 0 & 0\\
            0 & 0 & 0 & 1\\
            0 & 0 & 1 & 0\\
        \end{pmatrix}
        $
        &
        \makecell{\begin{tikzpicture}
            \begin{yquant}
                qubit {} q0;
                qubit {} q1;
                cnot q1 | q0;
            \end{yquant}
		\end{tikzpicture}} \\ 
        \hline
        \multirow{2}{*}{SWAP} 
        & \makecell{$\begin{pmatrix}
            1 & 0 & 0 & 0 \\
            0 & 0 & 1 & 0 \\
            0 & 1 & 0 & 0 \\
            0 & 0 & 0 & 1
        \end{pmatrix}$}
        & \makecell{
        \begin{tikzpicture}
            \begin{yquant}
                qubit {} q1;
                qubit {} q0;
                swap (q0,q1);
            \end{yquant}
        \end{tikzpicture}} \\
        \cline{2-3}
        & \makecell{$2n$-qubit SWAP gate is an operator that swaps  \\
        the qubits of the two $n$-qubit registers pairwise.}
        & \makecell{\begin{tikzpicture}
            \begin{yquant}
                qubit {} q1;
                qubit {} q0;
                ["north east:$n$" 
                {font=\protect\footnotesize, inner sep=0pt}]
                slash q1;
                ["north east:$n$" 
                {font=\protect\footnotesize, inner sep=0pt}]
                slash q0;
                hspace {5pt} -;
                box {\rm SWAP} (q0,q1);
                hspace {5pt} -;
            \end{yquant}
        \end{tikzpicture}} \\
        \hline
        Multiplexor operation & \makecell{$U = \sum_{l=0}^{2^m-1}\ket{l}\bra{l}_{\rm idx}\otimes U_l$, where $U_l$ are $n$-qubit \\ unitaries controlled by $m$-qubit register idx.} 
        & 
        \makecell{\begin{tikzpicture}
            \begin{yquant}
                qubit {idx} l;
                qubit {} j;
                ["north east:$m$" 
                {font=\protect\footnotesize, inner sep=0pt}]
                slash l;
                ["north east:$n$" 
                {font=\protect\footnotesize, inner sep=0pt}]
                slash j;
                hspace {5pt} -;
                [multictrl]
                box {$U$} j ~ l;
                hspace {5pt} -;
            \end{yquant}
        \end{tikzpicture}} \\ 
        \hline
        \makecell{Multiplexed $X$ gate \\ ($O$-control $X$ gate)} 
        & \makecell{$\sum\limits_{l\in \{O\}}\ket{l}\bra{l}_{\rm idx} \bigotimes_{k=1}^{n}[X]_{l,k}  + \sum\limits_{l'\notin \{O\}}\ket{l'}\bra{l'}_{\rm idx}\otimes I_{2^n}$,\\
        where $[X]_{l,k}$ is a Pauli-$X$ or identity operator \\ acting on the $k$th qubit.} & 
        \makecell{\begin{tikzpicture}
            \begin{yquant}
                qubit {idx} l;
                qubit {} 0;
                ["north east:$m$" 
                {font=\protect\footnotesize, inner sep=0pt}]
                slash l;
                ["north east:$n$" 
                {font=\protect\footnotesize, inner sep=0pt}]
                slash 0;
                hspace {5pt} -;
                [plusctrl, shape=yquant-circle]
                box {$O$} (l) | 0;
                hspace {5pt} -;
            \end{yquant}
        \end{tikzpicture}} \\
        \hline \cline{1-3}
    \end{tabular}
    \caption{Notations, mathematical forms, and circuits of basic controlled quantum gates.}
    \label{notation of quantum circuits}
\end{table}

In diagrams of quantum circuits, we follow standard convention in Table \ref{notation of quantum circuits}, which introduces some basic controlled  quantum gates used in this paper. The controlled gate is an important quantum gate that uses one or more qubits to control operation on other qubits. A horizontal line represents a single qubit, and a register consisting of multiple qubits is represented by adding a short slash at the beginning of a horizontal line. The rectangular box represents a single- or multi-qubit gate, and the circle represents the control set. For simplicity, if a Pauli-$X$ gate is controlled by one qubit, it is called a controlled $X$ (CNOT) gate.

\section{Block Encoding}
\label{section: block encoding}
In this section, we introduce the dictionary data structure for block encoding of sparse structured matrices that unifies multiple approaches. Using this data structure, we present a dictionary-based block encoding scheme of sparse matrices with low \textit{subnormalization}. We further extend this framework to its Hermitian counterpart for non-negative symmetric matrices and analyze the \textit{circuit depth} required for implementation.

\subsection{Dictionary Data Structure}
\label{subsection: Dictionary Data Structure}

A dictionary data structure is a collection of key-value pairs. It allows for the storage and retrieval of data based on a unique key, which is used to access the associated value. A sparse structured matrix is a type of matrix that combines two properties: sparsity (many zero elements) and structured non-zero elements (repetition and predictable arrangement). This type of matrix can be represented by a dictionary data structure, shown in Table \ref{dictionary}. The key is a non-negative integer. The associated value is a set consisting of triplets,
\begin{equation*}
    (\textit{element value}, \textit{row index}, \textit{column index}).
\end{equation*}
For triplets in the $l$-th key-value pair, they share the same value $A_l$, while the row and column indices $(i,j)$ belong to the coupled index set $(S_r(l), S_c(l))$. For different key-value pairs, the element values of triplets can be the same. 

\begin{table}[htbp]
    \centering
    \begin{tabular}{|c|c|}
        \hline
        \textbf{Keys} & \textbf{Values} \\
        \hline
        0 & $\left\{(a_{ij},i,j) : a_{ij} = A_0, (i,j)\in\left( S_r(0),S_c(0)\right)\right\}$ \\
        \hline
        1 & $\left\{(a_{ij},i,j) : a_{ij} = A_1, (i,j)\in \left(S_r(1), S_c(1)\right)\right\}$ \\
        \hline
        2 & $\left\{(a_{ij},i,j) : a_{ij} = A_2, (i,j)\in \left(S_r(2), S_c(2)\right)\right\}$ \\
        \hline
        \vdots & \vdots \\
        \hline
    \end{tabular}
    \caption{The dictionary data structure of a sparse structured matrix.}
    \label{dictionary}
\end{table}

The dictionary data structure of a sparse structured matrix is not unique. In the dictionary data structure framework, we define the following terminology.
\begin{itemize}
    \item A \textit{data item} refers to a key-value pair;
    \item The key $l$ is called the \textit{data index} and the associated value $\left\{(a_{ij},i,j) : a_{ij} = A_l, i\in S_r(l), j\in S_c(l)\right\}$ is called the \textit{data set};
    \item The non-zero value $A_l$ is called the \textit{data value};
    \item The index set $S_r(l)$ and $S_c(l)$ are called the \textit{row index set} and \textit{column index set}, respectively. They can be characterized by functions, which are collectively referred to as \textit{data functions}.
\end{itemize}

Our dictionary framework provides a unified approach to block encoding of sparse matrices, encompassing several existing protocols ~\cite{gilyen2019quantum,parker2024sfable,camps2024explicit,sunderhauf2024block} (see Appendix~\ref{subsection: several sparse block encodings within dictionary} for detailed comparisons). Inspired by~\cite{camps2024explicit}, in this work, we define the following dictionary data structure, which is shown in Table~\ref{table: dictionary data structure in this paper}. The row and column index sets can be characterized by the following data function:
\begin{itemize}
    \item[(1)] The row indices $i= c_l(j)$, where $c_l(j)$ is an injective data function that maps column indices $j$ to the corresponding row indices $i$ (unlike the bijective function $c(l,j)$ used in \cite{camps2024explicit}), $l\in[0,s_0-1]$;
    \item[(2)] The column indices $j\in S_c(l)$, where $S_c(l)$ denotes the set that includes all column indices of the $l$-th data value. 
\end{itemize}

\begin{table}[htbp]
    \centering
    \begin{tabular}{|c|c|}
        \hline
        \textbf{Keys} & \textbf{Values} \\
        \hline
        0 & $\left\{(a_{ij},i,j) : a_{ij} = A_0, (i,j)\in\left(c_{0}(j),S_{c}(0)\right) \right\}$ \\
        \hline
        1 & $\left\{(a_{ij},i,j) : a_{ij} = A_1, (i,j)\in\left(c_{1}(j),S_{c}(1)\right) \right\}$ \\
        \hline
        2 & $\left\{(a_{ij},i,j) : a_{ij} = A_2, (i,j)\in\left(c_{2}(j),S_{c}(2)\right) \right\}$ \\
        \hline
        \vdots & \vdots \\
        \hline
    \end{tabular}
    \caption{Dictionary data structure in this article.}
    \label{table: dictionary data structure in this paper}
\end{table}

In contrast to the block encoding protocol in~\cite{camps2024explicit} that mandates the presence of every data value across all columns, our block encoding based on the dictionary data structure in Table~\ref{table: dictionary data structure in this paper} eliminates such column-wise distribution constraints, thereby enabling more flexible handling of a wider variety of sparse matrices. 

Concrete examples of matrix representations using our dictionary data structure are presented in Section~\ref{section applications}, demonstrating its application to various sparse matrix types.

\subsection{Dictionary-based Block Encoding of Sparse Matrices}

In this section, we introduce the block encoding of sparse matrices and its Hermitian version based on the dictionary data structure in Tables \ref{table: dictionary data structure in this paper} and \ref{table: symmetric dictionary data structure in this paper}, respectively. The connection between the dictionary-based block encoding of sparse matrices and LCU is explored. 

The design rationale behind this encoding strategy is inspired by \cite{camps2024explicit} and \cite{sunderhauf2024block}. It is an improvement of the block encoding scheme in \cite{camps2024explicit} by removing the limitation that every data value should appear in all columns and using the state preparation technique proposed in \cite{sunderhauf2024block} to reduce its \textit{subnormalization}.

\subsubsection{Dictionary-based Block Encoding of Sparse Matrices and LCU}
Based on the dictionary data structure defined in Table~\ref{table: dictionary data structure in this paper}, our block encoding implementation consists of two quantum operations:
\begin{itemize}
    \item State Preparation: The oracles $\text{PREP}$ and $\text{UNPREP}$ perform the amplitude encoding of the data values $\{A_l\}_{l=0}^{s_0-1}$ stored in the dictionary;
    \item Index Mapping: The oracle $O_{\rm c}$ implements the injective function $c_l: j \mapsto c_l(j)$ for $l\in[0,s_0-1]$, $j\in S_c(l)$.
\end{itemize}
Leveraging these oracles, we formally present the dictionary-based block encoding scheme of sparse matrices in Theorem~\ref{theorem: sparse block encoding}.
\begin{figure}[htbp]
    \centering
    \begin{tikzpicture}
        \begin{yquant}
            qubit {idx} l;
            qubit {del} del;
            qubit {$\ket{j}$} j;
            ["north:$n$" 
            {font=\protect\footnotesize, inner sep=0pt}]
            slash j;
            ["north:$m$" 
            {font=\protect\footnotesize, inner sep=0pt}]
            slash l;
            hspace {5pt} -;
            box {PREP} l;
            hspace {5pt} -;
            [multictrl]
            box {$O_{c}$} (j,del) ~ l ;
            hspace {5pt} -;
            box {UNPREP} l;
            hspace {5pt} -;
            output {$\ket{i}$} j;
        \end{yquant}
    \end{tikzpicture}
    \caption{Basic framework of block encoding of sparse matrices with the dictionary data structure in Table~\ref{table: dictionary data structure in this paper}.}
    \label{circuit: sparse block encoding}
\end{figure}
\begin{theorem}\label{theorem: sparse block encoding}
    Let $A\in \mathbb{C}^{2^{n} \times 2^{n}}$ be a matrix that can be represented by a dictionary data structure with $s_{0}$ data items as stated in Table~\ref{table: dictionary data structure in this paper}, and $m=\lceil\log_{2}s_0\rceil$. If there exists a column index oracle $O_{c}$ such that 
    \begin{equation}\label{oracle Oc}
        O_{c}\ket{l}_{\rm idx} \ket{0}_{\rm del}\ket{j} = \begin{cases}
            \ket{l}_{\rm idx} \ket{0}_{\rm del} \ket{c_{l}(j)}, &\mbox{ if } l \in \left[0,s_{0}-1\right] \mbox{ and } j\in S_{c}\left(l\right),\\
            \ket{l}_{\rm idx} \ket{1}_{\rm del} \ket{j}, &\mbox{ if } l\in\left[s_{0},2^m-1\right] \mbox{ or } j\notin S_{c}\left(l\right),
        \end{cases}
    \end{equation}
    and two state preparation oracles $\rm PREP,UNPREP$ such that 
    \begin{equation}\label{oracle: PREP}
        {\rm PREP}\ket{0}_{\rm idx}^{\otimes m} = \frac{1}{\sqrt{\sum_{l=0}^{s_{0}-1}\left|A_{l}\right|}}\left(\sum_{l=0}^{s_{0}-1}\sqrt{A_{l}}\ket{l}_{\rm idx} + \sum_{l=s_{0}}^{2^m-1}0\ket{l}_{\rm idx}\right),
    \end{equation}
    \begin{equation}\label{oracle: UNPREP}
        {\rm UNPREP}^{\dagger}\ket{0}_{\rm idx}^{\otimes m} = \frac{1}{\sqrt{\sum_{l=0}^{s_{0}-1}\left|A_{l}\right|}}\left(\sum_{l=0}^{s_{0}-1}\sqrt{A_{l}}^{*}\ket{l}_{\rm idx} + \sum_{l=s_{0}}^{2^m-1}0\ket{l}_{\rm idx}\right),
    \end{equation}
    where $\sqrt{A_{l}}^{*}$ denotes the complex conjugate of $\sqrt{A_{l}}$, then the unitary
    \begin{equation*}
        U_{A} = \left({\rm UNPREP} \otimes I_{2^{n+1}} \right) O_{c} \left({\rm PREP} \otimes I_{2^{n+1}} \right),
    \end{equation*}
    as shown in Figure \ref{circuit: sparse block encoding}, can block encode $A$ with the subnormalization $\alpha = \sum_{l=0}^{s_{0}-1}\left|A_{l}\right|$.
\end{theorem}
\begin{proof}
     To recover the matrix from the block encoding, the index register and the del register must be initialized and postselected as $\ket{0}$. The values $\frac{1}{\sum_{l=0}^{s_0-1}\vert A_l\vert}a_{ij}$ are then recovered when initializing the bottom register with $\ket{j}$ and postselecting/measuring an $\ket{i}$ as
    \begin{equation*}
    \begin{aligned}
        & \bra{0}_{\rm idx}^{\otimes m}\bra{0}_{\rm del}\bra{i} \left({\rm UNPREP} \otimes I_{2^{n+1}} \right) O_{c} \left({\rm PREP} \otimes I_{2^{n+1}} \right) \ket{0}_{\rm idx}^{\otimes m}\ket{0}_{\rm del} \ket{j}\\
        =& \frac{1}{\sqrt{\sum_{l'=0}^{s_{0}-1}\left|A_{l'}\right|\sum_{l=0}^{s_{0}-1}\left|A_{l}\right|}}\sum_{l',l=0}^{s_{0}-1} \bra{l'}_{\rm idx}\bra{0}_{\rm del}\bra{i} \sqrt{A_{l'}A_{l}} O_{c} \ket{l}_{\rm idx}\ket{0}_{\rm del}\ket{j} \\
        =& \frac{1}{\sqrt{\sum_{l'=0}^{s_{0}-1}\left|A_{l'}\right|\sum_{l=0}^{s_{0}-1}\left|A_{l}\right|}}\sum_{l',l=0}^{s_{0}-1} \bra{l'}_{\rm idx}\bra{i} \sqrt{A_{l'}A_{l}} \delta_{l,[0,s_{0}-1]}\delta_{j,S_{c}(l)}\ket{l}_{\rm idx}\ket{c_l(j)} \\
        =& \frac{1}{\sqrt{\sum_{l'=0}^{s_{0}-1}\left|A_{l'}\right|\sum_{l=0}^{s_{0}-1}\left|A_{l}\right|}} \sum_{l',l=0}^{s_{0}-1} \sqrt{A_{l'}A_{l}} \delta_{l,[0,s_{0}-1]}\delta_{j,S_{c}(l)} \delta_{l',l}\delta_{i,c_{l}(j)} \\
        =& \frac{1}{\sum_{l=0}^{s_{0}-1}\left|A_{l}\right|} a_{ij},
    \end{aligned}
    \end{equation*}
    where $a_{ij}$ refers to the $l$-th data value in the dictionary as $\left\{(a_{ij},i,j):a_{ij}=A_l,(i,j) \in \left(c_l(j),S_c(l)\right)\right\}$, and $\delta_{i,S}$ is the indicate function as 
    \begin{equation*}
        \delta_{i,S} = 
        \begin{cases}
            1, &\text{if }  i \in S \text{ or } i=S,\\
            0, &\text{otherwise.}
        \end{cases}
    \end{equation*}
\end{proof}

\begin{remark}
    We will present a low-depth circuit implementation of the aforementioned block encoding of sparse matrices in Section \ref{section: low time metric implementation}, achieving the \textit{subnormalization} as $\sum_l\vert A_l\vert$. Notably, while the theoretical lower bound for \textit{subnormalization} is given by the spectral norm $\Vert A\Vert_{\rm sp}$, the construction of a corresponding low-depth block encoding protocol remains an open challenge in this field.
\end{remark}

Another general approach for encoding sparse Hamiltonians is LCU~\cite{LongGui-Lu_2006,childs2012hamiltonian}, which constructs quantum circuits in the form,
\begin{equation}
    U = \sum_{l=0}^{L-1} c_l U_l,
    \label{eq:LCU}
\end{equation}
where $c_l > 0$ are positive coefficients satisfying $\sum_{l=0}^{L-1} c_l = 1$, $U_l$ are $n$-qubit unitary operators implementable with polynomial-size circuits, $L = \mathcal{O}(\mathrm{poly}(n))$ scales polynomially with system size $n$. For practical implementation with low \textit{circuit depth}, it is commonly assumed that the unitaries $U_l$ are self-inverse operators ($U_l^2 = I$)~\cite{babbush2018encoding}.
\begin{figure}[htbp]
    \centering
    \begin{tikzpicture}
        \begin{yquant}
            qubit {idx} l;
            qubit {} j;
            ["north:$n$" 
            {font=\protect\footnotesize, inner sep=0pt}]
            slash j;
            ["north:$\lceil\log {L}\rceil$" 
            {font=\protect\footnotesize, inner sep=0pt}]
            slash l;
            hspace {5pt} -;
            box {PREP} l;
            hspace {5pt} -;
            [multictrl] box {$\text{SELECT}$} j ~ l;
            hspace {5pt} -;
            box {UNPREP} l;
            hspace {5pt} -;
        \end{yquant}
    \end{tikzpicture}
    \caption{Quantum circuit of LCU~\cite{babbush2018encoding}.}
    \label{circuit: LCU}
\end{figure}

The quantum circuit for LCU is constructed in Figure~\ref{circuit: LCU}. Its oracles PREP and UNPREP are used to prepare quantum states, that is,
\begin{equation*}
    \text{PREP}\ket{0}^{\otimes \lceil\log {L}\rceil}=\text{UNPREP}^\dagger\ket{0}^{\otimes \lceil\log {L}\rceil}=\frac{1}{\sqrt{\sum_{l} c_l }}\sum_{l}\sqrt{c_l}\ket{l},
\end{equation*}
and the $\text{SELECT}$ is the Hamiltonian selection oracle~\cite{babbush2018encoding} defined as
$$
\begin{aligned}
    \text{SELECT} &= \sum_l \ket{l}\bra{l}\otimes U_l, \\
\end{aligned}
$$
where $U_l$ are $n$-qubit unitaries.

The following corollary establishes the connection between the dictionary-based block encoding of sparse matrices and LCU. 
\begin{corollary}\label{corollary: LCU}
Let $A\in \mathbb{C}^{2^{n} \times 2^{n}}$ be a matrix that can be represented by a dictionary data structure with $s_{0}$ data items as stated in Table~\ref{table: dictionary data structure in this paper}, and $m=\lceil\log_{2}s_0\rceil$. Then the dictionary-based block encoding of $A$ represents a linear combination of $(n+1)$-qubit unitaries.
\end{corollary}
\begin{proof}
    The dictionary-based block encoding of sparse matrices consists of three oracles: $O_{c}$, $\text{PREP}$, and $\text{UNPREP}$.

    \begin{itemize}
        \item The oracle $O_{c}$ in Equation~\eqref{oracle Oc} can also be expressed as 
        \begin{equation*}
            \begin{aligned}
            O_{\rm c} = \sum_{l=0}^{2^m-1}\ket{l}\bra{l}_{\rm idx} \otimes U_{l}^{(X)},
            \end{aligned}
        \label{eq_Org_Oc}
        \end{equation*}
        where $U_{l}^{(X)}$ are $(n+1)$-qubit unitaries as
        $$
        U_{l}^{(X)}= \bigotimes_{k=0}^{n-1} \left(\left[\text{XOR}\left([c_l(j)]_k,[j]_k\right)\right]X + \overline{\left[\text{XOR}\left([c_l(j)]_k,[j]_k\right)\right]} I \right),
        $$
        satisfying 
        \begin{equation*}
        \begin{aligned}
            U_{l}^{(X)} \ket{0}_{\rm del}\ket{j} &= \begin{cases}
                \ket{0}_{\rm del} \ket{c_{l}(j)}, &\mbox{if } l\in\left[0,s_{0}-1\right] \mbox{ and }j\in S_{c}(l),\\
                \ket{1}_{\rm del}\ket{j} , &\mbox{if } l\in\left[s_{0},2^m-1\right] \mbox{ or }j\notin S_{c}(l),\\
            \end{cases}
        \end{aligned}
        \end{equation*}
        XOR denotes the Exclusive Or operation of two $1$-bit binary operation. Therefore, the oracle $O_{\rm c}$ forms an $(n+1)$-qubit Hamiltonian simulation selection operator.

        \item The oracles $\text{PREP}$ and $\text{UNPREP}$ in Equation \eqref{oracle: PREP} and \eqref{oracle: UNPREP} encode the quantum states $\frac{1}{\sqrt{\sum_{l} \left|A_l\right|}}\sum_{l}\sqrt{A_l}\ket{l}$ and $\frac{1}{\sqrt{\sum_{l} \left|A_l\right|}}\sum_{l}\sqrt{A_l}^{*}\ket{l}$, respectively.
    \end{itemize}
    
    These oracles in the dictionary-based block encoding of sparse matrices shown in Figure \ref{circuit: Ul} correspond to the oracles with the same notations of LCU as shown in Figure~\ref{circuit: LCU}. Above all, the dictionary-based block encoding of $A\in\mathbb{C}^{2^n\times 2^n}$ with $s_0$ data items is a linear combination of $(n+1)$-qubit unitaries.

    \begin{figure}[htbp]
    \centering
    \begin{tikzpicture}
        \begin{yquant}
            qubit {idx} l;
            qubit {del} del;
            qubit {$\ket{j}$} j;
            
            ["north east:$n$" {font=\protect\footnotesize, inner sep=0pt}]
            slash j;
            ["north east:$m$" {font=\protect\footnotesize, inner sep=0pt}]
            slash l;
            hspace {5pt} -;
            box {PREP} l;
            hspace {5pt} -;
            [this subcircuit box style={dashed, "SELECT" below}]
            subcircuit {
                qubit {} l;
                qubit {} j;
                qubit {} del;
                [plusctrl, shape=yquant-circle]
                box {$O_{\rm c}$} l | j,del;
            } (l,del,j);
            hspace {5pt} -;
            box {UNPREP} l;
            hspace {5pt} -;
            output {$\ket{i}$} j;
        \end{yquant}
    \end{tikzpicture}
    \caption{Dictionary-based block encoding of $n$-qubit sparse matrices is a case of linear combination of $(n+1)$-qubit unitaries, where $O_c = \sum_l\ket{l}\bra{l}\otimes U_{l}^{(X)}$ forms a Hamiltonian selection oracle.}
    \label{circuit: Ul}
    \end{figure}
\end{proof}

\subsubsection{Dictionary-based Hermitian Block Encoding of Non-negative Symmetric Matrices}
Given that the data functions $c_l(j)$ are injective in $j$ for all $l \in [0, s_0 - 1]$, it is logical that we can construct the corresponding functions $c_j(l)$ for $j \in [0, 2^n - 1]$, which are injective in $l$. These functions map the data indices $l \in [0, s_0 - 1]$ to the row indices, establishing the dictionary data structure presented in Table~\ref{table: symmetric dictionary data structure in this paper}. 

\begin{table}[h!]
    \centering
    \begin{tabular}{|c|c|}
        \hline
        \textbf{Keys} & \textbf{Values} \\
        \hline
        0 & $\left\{(a_{ij},i,j) : a_{ij} = A_0, (i,j)\in\left(c_{j}(0),S_{c}(0)\right) \right\}$ \\
        \hline
        1 & $\left\{(a_{ij},i,j) : a_{ij} = A_1, (i,j)\in\left(c_{j}(1),S_{c}(1)\right) \right\}$ \\
        \hline
        2 & $\left\{(a_{ij},i,j) : a_{ij} = A_2, (i,j)\in\left(c_{j}(2),S_{c}(2)\right) \right\}$ \\
        \hline
        \vdots & \vdots \\
        \hline
    \end{tabular}
    \caption{Dictionary data structure for the Hermitian block encoding of non-negative symmetric matrices, where the data function $i=c_j(l)$ is different from the data function $i=c_l(j)$ stated in Table~\ref{table: dictionary data structure in this paper}.}
    \label{table: symmetric dictionary data structure in this paper}
\end{table}

Therefore, we propose the following theorem for the dictionary-based Hermitian block encoding.
\begin{theorem}\label{theorem: sparse Hermitian block encoding}
    Let $A\in \mathbb{R}^{2^n\times 2^n}$ be a non-negative symmetric matrix that can be represented by a dictionary data structure with $s_{0}$ data items as stated in Table~\ref{table: symmetric dictionary data structure in this paper}, and $m = \lceil\log_{2}s_0\rceil$, where $s_0 \leq 2^{n}$. If there exists a column index oracle $O_{c}$ such that 
    \begin{equation}\label{Hermitian oracle Oc}
        O_{c}\ket{0}^{\otimes (n-m)}_{\rm idx}\ket{l}_{\rm idx}\ket{0}_{\rm del} \ket{j} = \begin{cases}
            \ket{c_{j}(l)}_{\rm idx}\ket{0}_{\rm del} \ket{j}, &\mbox{ if } j\in S_{c}\left(l\right) \mbox{ and } l\in \left[0,s_{0}-1\right],\\
            \ket{0}^{\otimes (n-m)}_{\rm idx}\ket{l}_{\rm idx}\ket{1}_{\rm del} \ket{j}, &\mbox{ if } j\notin S_{c}\left(l\right) \mbox{ or } l\in \left[s_{0},2^m-1\right],
        \end{cases}
    \end{equation}
    and a state preparation oracle $\rm PREP$ such that
    \begin{equation*}
        {\rm PREP}\ket{0}^{\otimes m}_{\rm idx} = \frac{1}{\sqrt{\sum_{l=0}^{s_{0}-1}A_{l}}}\left(\sum_{l=0}^{s_{0}-1}\sqrt{A_{l}}\ket{l}_{\rm idx} + \sum_{l=s_{0}}^{2^m-1}0\ket{l}_{\rm idx}\right),
    \end{equation*}
    then $U_{A}$ represented by the circuit shown in Figure \ref{circuit: sparse Hermitian block encoding} is a Hermitian block encoding of $A$ with the \textit{subnormalization} $\sum_{l=0}^{s_{0}-1}A_{l}$. 
    \begin{figure}[htbp]
        \centering
        \begin{tikzpicture}
            \begin{yquant}
                qubit {} l1;
                qubit {} l0;
                qubit {del1} del1;
                qubit {del0} del0;
                qubit {$\ket{j}$} j;
                init {idx} (l0, l1);

                ["north:$n-m$" {font=\protect\footnotesize, inner sep=0pt}]
                slash l1;
                ["north east:$m$" {font=\protect\footnotesize, inner sep=0pt}]
                slash l0;
                ["north east:$n$" {font=\protect\footnotesize, inner sep=0pt}]
                slash j;
                hspace {5pt} -;
                box {PREP} l0;
                hspace {5pt} -;
                [multictrl]
                box {$O_{c}$} (l1,l0,del1) ~ j;
                hspace {5pt} -;
                swap (del1,del0);
                box {SWAP} (l1,l0,j);
                hspace {5pt} -;
                [multictrl]
                box {$O_{c}^{\dagger}$} (l1,l0,del1) ~ j;
                hspace {5pt} -;
                box {${\rm PREP}^{\dagger}$} l0;
                hspace {5pt} -;
                output {$\ket{i}$} j;
            \end{yquant}
        \end{tikzpicture}
        \caption{Basic framework of Hermitian block encoding of non-negative symmetric matrices with the dictionary data structure in Table~\ref{table: symmetric dictionary data structure in this paper}. A wavy line connecting two boxes indicates that the operator acts on the registers associated with those two boxes.}
        \label{circuit: sparse Hermitian block encoding}
    \end{figure}
\end{theorem}
\begin{proof}
Performing $U_A$ on $\ket{0}^{\otimes n}_{\rm idx}\ket{0}_{\rm del1}\ket{0}_{\rm del0}\ket{j}$ and postselecting/measuing it with $\bra{0}^{\otimes n}_{\rm idx}\bra{0}_{\rm del1}\bra{0}_{\rm del0}\bra{i}$, we obtain
\begin{equation*}
    \begin{aligned}
        &\bra{0}^{\otimes n}_{\rm idx}\bra{0}_{\rm del1}\bra{0}_{\rm del0}\bra{i}U_{A}\ket{0}^{\otimes n}_{\rm idx}\ket{0}_{\rm del1}\ket{0}_{\rm del0}\ket{j} \\
        =& \frac{1}{\sum_{l=0}^{s_0-1}A_l} \sum_{l,l'=0}^{2^m-1} \bra{0}^{\otimes (n-m)}_{\rm idx}\bra{l'}_{\rm idx}\bra{0}_{\rm del1}\bra{0}_{\rm del0}\bra{i} \sqrt{A_{l'}A_l} O_c^{\dagger} {\rm SWAP} O_c \ket{0}^{\otimes (n-m)}_{\rm idx}\ket{l}_{\rm idx}\ket{0}_{\rm del1}\ket{0}_{\rm del0}\ket{j} \\
        =& \frac{1}{\sum_{l=0}^{s_0-1}A_l} \sum_{l,l'=0}^{2^m-1} \delta_{i,S_c(l')}\delta_{l',[0,s_0-1]}\bra{c_{i}(l')}_{\rm idx}\bra{i} \sqrt{A_{l'}A_l} \delta_{j,S_c(l)}\delta_{l,[0,s_0-1]}\ket{j}_{\rm idx}\ket{c_j(l)}\\
        =& \frac{1}{\sum_{l=0}^{s_{0}-1}A_{l}} \sum_{l',l=0}^{2^m-1} \sqrt{A_{l'}A_{l}} (\delta_{i,S_c(l')}\delta_{l',[0,s_0-1]}  \delta_{j,c_{i}(l')}) (\delta_{j,S_c(l)}\delta_{l,[0,s_0-1]} \delta_{i, c_{j}(l)}) \\
        =& \frac{1}{\sum_{l=0}^{s_{0}-1}A_{l}} \sqrt{a_{ji} a_{ij}} \\
        =& \frac{1}{\sum_{l=0}^{s_{0}-1}A_{l}}a_{ij},
    \end{aligned}
\end{equation*}
where $a_{ij}$ refers to the $l$-th data value in the dictionary as $\left\{(a_{ij},i,j) : a_{ij} = A_l, (i,j)\in\left(c_{j}(l),S_{c}(l)\right) \right\}$.
\end{proof}

Note that the oracle $O_{c}$ in Equation \eqref{Hermitian oracle Oc} is different from that in Equation \eqref{oracle Oc}. The former is controlled by the register ${\rm idx}$ and evolves the register $\ket{j}$, whereas the latter is controlled by the register $\ket{j}$ and evolves the register ${\rm idx}$. These two oracles $O_{c}$ are both reasonable because the functions $c_{l}(j)$ and $c_{j}(l)$ are separately injective in both $j$ and $l$.

\section{Low Time Metric Implementation}\label{section: low time metric implementation}

The \textit{time metric} of a block-encoding is defined in Equation~\eqref{time metric} as
\begin{equation*}
    \textit{time metric} = \textit{circuit depth} \times \textit{subnormalization}.
\end{equation*}
There is a trade-off between the \textit{circuit depth}, ancillary qubits, and \textit{subnormalization} in the dictionary-based block encoding of sparse matrices. To simplify the statement of implementation and proof, we make the following assumptions on the sparse matrix $A\in\mathbb{C}^{2^n\times 2^n}$,
\begin{assumption}
    \item \label{assumption: 1} Suppose that the encoding dictionary (as specified in Table~\ref{table: dictionary data structure in this paper}) has $s_0$ data items, where $\lceil \log  s_0 \rceil = \mathcal{O}(n)$;
    \item \label{assumption: 2} Suppose that the matrix $A$ has a total of $s$ non-zero elements.
\end{assumption}

\subsection{Low Circuit Depth Implementations of Oracles}
In this section, we analyze the oracle implementation of our block encoding using the $\left\{\mbox{U}(2),\mbox{CNOT}\right\}$ gate set. The construction depends on efficient implementations of several key oracles, including the column oracle $O_{c}$ and the state preparation oracles $\text{PREP}$ and $\text{UNPERP}$.

\subsubsection{Implementation of $O_{c}$}

The column oracle $O_{c}$ can be efficiently implemented by adapting the low-depth circuit design of the SAIM~\cite{zhang2024circuit}, an advanced framework for representing general sparse matrices, and the LCU~\cite{berry2012black,PhysRevX.13.041041}. This implementation leverages the sparse Boolean memory architecture~\cite{zhang2024circuit,zhang2022quantum} to query the positions of non-zero elements, which serves as the foundation for the low-depth circuit realization.

\begin{definition}[Sparse Boolean Memory
(SBM)~\cite{zhang2022quantum}]
    For an $n$-index, $\tilde{n}$-word Boolean function $f:\{0,1\}^{n}\to\{0,1\}^{\tilde{n}}$, let $\mathcal{S}_f=\{k:f(k)\neq 0\cdots0\}$ be a set that contains all input indices with non-zero output. We say that $f$ is $s$-sparse if $\mathcal{S}_f$ has no more than $s$ elements. Its corresponding sparse Boolean function selector satisfies $\text{select}(f)\ket{k}\ket{z}=\ket{k}\ket{z\oplus f(x)}$, where $\oplus$ represents bitwise XOR. Let $\left[f(k)\right]_l$ be the $l$-th digit of $f(k)$, $\text{select}(f)$ can also be expressed as
    \begin{equation}
        \mbox{select}(f)\equiv \sum_{k=0}^{2^n-1}\ket{k}\bra{k}\bigotimes_{l=0}^{\tilde{n}-1}\left([f(k)]_lX + \overline{[f(k)]}_lI\right).
    \label{equation: select(f)}
    \end{equation}
\end{definition}

\begin{lemma}[Circuit Depth of SBM~\cite{zhang2022quantum}]\label{SBM}
    Given an arbitrary $n$-index, $\tilde{n}$-word, $s$-sparse Boolean function $f$, $\text{select}(f)$ in Equation \eqref{equation: select(f)} can be realized with \textit{circuit depth} $\mathcal{O}\left(\log(ns\tilde{n})\right)$ and $\mathcal{O}\left(ns\tilde{n}\right)$ ancillary qubits using only single- and two-qubit gates.
\end{lemma}

Based on SBM, the column oracle $O_{\rm c}$ in Equation~\eqref{oracle Oc} can then be implemented in the following five stages.  

\begin{itemize}
    \item Stage 1. Perform a Pauli-$X$ on the del register.
    \begin{equation}\label{eq_O_c0}
        \ket{0}_{\rm del} \xrightarrow{X} \ket{1}_{\rm del}.
    \end{equation}
    
    \item Stage 2. Perform a ranging SBM $O_{c_1}$.

    We introduce a $(\lceil \log s_0\rceil+n)$-index, $1$-word, $s_{\rm c_1}$-sparse Boolean function $f_{c_1}(l,j):\{0,1\}^{\lceil \log s_0\rceil+n} \mapsto \{0,1\}$ defined as
    \begin{equation*}
        f_{c_1}(l,j) = \left\{\begin{aligned}
        &1, && \mbox{if } l\in\left[0,s_{0}-1\right] \mbox{ and }j\in S_{c}(l), \\
        &0, && \mbox{if } l\in\left[s_{0},2^{\lceil \log s_0\rceil}-1\right] \mbox{ or }j\notin S_{c}(l).
        \end{aligned}\right.
    \end{equation*}
     The sparsity of $f_{\rm c_1}$ is given by
    \begin{equation*}
        \begin{aligned}
            s_{\rm c_1} &= \left|\left\{(l,j): f_{\rm c_1 }(l,j)\neq 0\right\}\right|  = \left|\left\{(l,j):l\in[0,s_{0}-1]\mbox{ and }j\in S_{c}(l)\right\}\right| = s,
        \end{aligned}
    \end{equation*}
    where the last equality holds since the number of non-zero element's indices $(l,j)$ is equal to the count of non-zero elements in the matrix. Therefore, there exists a SBM $O_{\rm c_1}$ such that
    \begin{equation}\label{eq_O_c1}
        O_{\rm c_1} = \mbox{select}(f_{\rm c_1})= \sum_{l,j=0}^{2^{\lceil \log s_0\rceil}-1,2^n-1}\ket{l}\bra{l}_{\rm idx} \otimes \left(f_{\rm c_1}(l,j)X + \overline{f_{\rm c_1}(l,j)} I\right)_{\rm del} \otimes \ket{j}\bra{j}.
    \end{equation}

    \item Stage 3. Perform a mapping SBM $O_{\rm c_2}$.

    We introduce a $(\lceil \log s_0\rceil+n+1)$-index, $n$-word, $s_{\rm c_2}$-sparse Boolean function $f_{\rm c_2}(l,j):\{0,1\}^{\lceil \log s_0\rceil+n+1}\mapsto \{0,1\}^{n}$ as
    \begin{equation*}
        f_{\rm c_2}(l,j,\text{del}) = \left\{\begin{aligned}
        & c_l(j), &&\text{if } \text{del} = 0,\\
        & 0, &&\text{if } \text{del} = 1,\\
        \end{aligned}\right.
    \end{equation*}
    where the sparsity of $f_{\rm c_2}$ is given by $s_{\rm c_2} = \left|\left\{(l,j):f_{\rm c_2}(l,j)\neq 0\right\}\right| = s_{c_1} = s$. Furthermore, there exists a SBM $O_{\rm c_2}$ such that 
    \begin{equation}\label{eq_O_c2}
        \begin{aligned}
        O_{\rm c_2} =& \mbox{select}(f_{\rm c_2}) \\
        =& \sum_{l=0,j=0}^{2^{\lceil \log s_0\rceil}-1,2^n-1}\ket{l}\bra{l}_{\rm idx}\otimes \ket{0}\bra{0}_{\rm del}\otimes \ket{j}\bra{j}\otimes  \bigotimes_{k=0}^{n-1}\left([f_{\rm {c_2}}(l,j, 0)]_{k}X + \overline{[f_{\rm {c_2}}(l, j, 0)]}_{k}I\right) \\
        &+ \sum_{l=0,j=0}^{2^{\lceil \log s_0\rceil}-1,2^n-1}\ket{l}\bra{l}_{\rm idx}\otimes \ket{1}\bra{1}_{\rm del}\otimes \ket{j}\bra{j}\otimes  \bigotimes_{k=0}^{n-1}\left([f_{\rm {c_2}}(l,j, 1)]_{k}X + \overline{[f_{\rm {c_2}}(l, j, 1)]}_{k}I\right).
        \end{aligned}        
    \end{equation}
    
    \item Stage 4. Perform an uncompute SBM $O_{\rm c_3}$.

    Due to the dictionary requirement of Table \ref{table: dictionary data structure in this paper}, $c_l(j)$ is an injective function. Given $l$ and $j$, there exists a $(\lceil \log s_0\rceil+n+1)$-index, $n$-word, $s_{\rm c_3}$-sparse Boolean function $f_{\rm c_3}$ such that
    \begin{equation*}
        f_{\rm c_3}(l,c_l(j),\text{del}) = \begin{cases}
            j, &\text{if del}=0,\\
            0, &\text{if del}=1,
        \end{cases}
    \end{equation*}
    where the sparsity of $f_{\rm c_3}$ is given by $s_{\rm c_3} = \left|\left\{(l,j):f_{\rm c_3}(l,j,\text{del})\neq 0\right\}\right| = s_{\rm c_1} = s$. Similarly, there exists a SBM $O_{\rm c_3}$ such that 
    \begin{equation}\label{eq_O_c3}
        \begin{aligned}
            O_{\rm c_3}
            =& \mbox{select}(\hat{U}_{c_3}) \\
            =& \sum_{l=0,j=0}^{2^{\lceil \log (s_0)\rceil}-1,2^n-1}\ket{l}\bra{l}_{\rm idx}\otimes \ket{0}\bra{0}_{\rm del} \otimes \left[\bigotimes_{k=0}^{n-1}\left([f_{\rm c_3}(l,c_l(j),0)]_{k}X + \overline{[f_{\rm c_3}(l,c_l(j),0)]}_{k}I\right)\right]\\
            & \qquad\qquad\qquad\qquad\quad\quad\qquad \otimes \ket{c_l(j)}\bra{c_l(j)}, \\
            & + \sum_{l=0,j=0}^{2^{\lceil \log (s_0)\rceil}-1,2^n-1}\ket{l}\bra{l}_{\rm idx}\otimes \ket{1}\bra{1}_{\rm del} \otimes \left[\bigotimes_{k=0}^{n-1}\left([f_{\rm c_3}(l,0,1)]_{k}X + \overline{[f_{\rm c_3}(l,0,1)]}_{k}I\right)\right]\\
            & \qquad\qquad\qquad\qquad\quad\quad\qquad \otimes \ket{0}\bra{0}^{\otimes n}. \\
        \end{aligned}
    \end{equation}
    
    \item Stage 5. Perform a $0$-controlled $2n$-qubit SWAP gate $O_{\rm SWAP}$ gate.
    \begin{equation}\label{eq O_bit}
        \ket{0}_{\rm del}\ket{0}^{\otimes n}\ket{c_l(j)} 
        \xrightarrow{O_{\rm SWAP}} \ket{0}_{\rm del}\ket{c_l(j)} \ket{0}^{\otimes n},\quad 
        \ket{1}_{\rm del}\ket{j}\ket{0}^{\otimes n} 
        \xrightarrow{O_{\rm SWAP}} \ket{1}_{\rm del} \ket{j} \ket{0}^{\otimes n} .
    \end{equation}
\end{itemize}
Above all, the oracle $O_{\rm c}$ can be represented by
\begin{equation*}
    O_{\rm c} = O_{\rm SWAP}O_{\rm c_3}O_{\rm c_2}O_{\rm c_1}X.
    \label{eq_SBM_Oc}
\end{equation*}
The implementation of oracle $O_{c}$ leads to the following process,
$$
\begin{aligned}
    \ket{l}_{\rm idx}\ket{0}_{\rm del}\ket{j} \ket{0}^{\otimes n} &\xrightarrow{X} \ket{l}_{\rm idx}\ket{1}_{\rm del}\ket{j} \ket{0}^{\otimes n} \xrightarrow{O_{\rm c_1}} \ket{l}_{\rm idx}\ket{1\oplus f_{\rm c_1}(l,j)}_{\rm del}\ket{j}\ket{0}^{\otimes n} ,
\end{aligned}
$$
if $l\in[0,s_0-1]$ and $j\in S_{c}(l)$, we have
\begin{equation*}
    \begin{aligned} 
        & \ket{l}_{\rm idx}\ket{1\oplus f_{\rm c_1}(l,j)}_{\rm del}\ket{j}\ket{0}^{\otimes n} = \ket{l}_{\rm idx}\ket{0}_{\rm del}\ket{j} \ket{0}^{\otimes n} \\
        \xrightarrow{O_{\rm c_2}} & \ket{l}_{\rm idx}\ket{0}_{\rm del}\ket{j}\ket{0 \oplus f_{c_2}(l,j,0) } = \ket{l}_{\rm idx}\ket{0}_{\rm del}\ket{j}\ket{c_l(j)} \\
        \xrightarrow{O_{\rm c_3}} & \ket{l}_{\rm idx}\ket{0}_{\rm del}\ket{j\oplus f_{c_3}\left(l,c_l(j),0\right)} \ket{c_l(j)} = \ket{l}_{\rm idx}\ket{0}_{\rm del}\ket{j\oplus j}\ket{c_l(j)} \\
        \xrightarrow{O_{\rm SWAP}} & \ket{l}_{\rm idx}\ket{0}_{\rm del}\ket{c_l(j)}\ket{0}^{\otimes n};
    \end{aligned} 
\end{equation*}
if $l\in[s_0,2^m-1]$ or $j\notin S_c(l)$, we have
\begin{equation*}
    \begin{aligned} 
        & \ket{l}_{\rm idx}\ket{1\oplus f_{\rm c_1}(l,j)}_{\rm del}\ket{j}\ket{0}^{\otimes n} = \ket{l}_{\rm idx}\ket{1}_{\rm del}\ket{j} \ket{0}^{\otimes n} \\
        \xrightarrow{O_{\rm c_2}} & \ket{l}_{\rm idx}\ket{1}_{\rm del}\ket{j}\ket{0 \oplus 0 } = \ket{l}_{\rm idx}\ket{1}_{\rm del}\ket{j}\ket{0}^{\otimes n} \\
        \xrightarrow{O_{\rm c_3}} & \ket{l}_{\rm idx}\ket{1}_{\rm del}\ket{j\oplus 0} \ket{0}^{\otimes n} = \ket{l}_{\rm idx}\ket{1}_{\rm del}\ket{j}\ket{0}^{\otimes n} \\
        \xrightarrow{O_{\rm SWAP}} & \ket{l}_{\rm idx}\ket{1}_{\rm del}\ket{j}\ket{0}^{\otimes n}.
    \end{aligned} 
\end{equation*}
That is, 
\begin{equation*}
    \ket{l}_{\rm idx}\ket{0}_{\rm del}\ket{j}\ket{0}^{\otimes n} \xrightarrow{O_{\rm c}} \begin{cases}
        \ket{l}_{\rm idx}\ket{0}_{\rm del}\ket{c_l(j)}\ket{0}^{\otimes n}, &\text{ if }l\in[0,s_0-1]\text{ and }j\in S_{c}(l),\\
        \ket{l}_{\rm idx}\ket{1}_{\rm del}\ket{j}\ket{0}^{\otimes n}, &\text{ if }l\in[s_0,2^m-1]\text{ or }j\notin S_{c}(l).\\
    \end{cases}
\end{equation*}

\subsubsection{Implementation of $\text{PREP}$ and $\text{UNPREP}$}
There are two oracles $\text{PREP}$ and $\text{UNPREP}$ for state preparation. The task of quantum state preparation is to prepare an $n$-qubit quantum state $\ket{\psi}$ from an initial product state $\ket{0}^{\otimes n}$ using single- and two-qubit gates. A general quantum state can be expressed as
\begin{equation*}
    \ket{\psi} = \sum_{k=0}^{N-1} \psi_k\ket{k},
\end{equation*}
where $N = 2^n$, $\psi_k \in \mathbb{C}$, $\sum_{k=0}^{N-1} |\psi_k|^2 = 1$.

\begin{lemma}[State Preparation~\cite{zhang2022quantum}]\label{State Preparation}
    With only single- and two-qubit gates, an arbitrary $n$-qubit quantum state can be deterministically prepared with circuit depth $\mathcal{O}\left(n\right)$ and $\mathcal{O}\left(2^n\right)$ ancillary qubits.
\end{lemma}

\begin{lemma}[Sparse State Preparation~\cite{zhang2022quantum}]\label{Sparse State Preparation}
    With only single- and two-qubit gates,  arbitrary $n$-qubit, $d$-sparse ($d\ge 2$) quantum states can be deterministically prepared with circuit depth $\Theta\left(\log(nd)\right)$ and $\mathcal{O}\left(nd\log d\right)$ ancillary qubits.
\end{lemma}

\subsection{Proof of Circuit Depth}
The following theorem demonstrates the low \textit{circuit depth} for implementing our block encoding protocol of sparse matrices based on the dictionary data structure specified in Table~\ref{table: dictionary data structure in this paper}.

\begin{theorem}[Circuit depth of dictionary-based block encoding of sparse matrices]\label{theorem: circuit depth}
    Given a sparse matrix $A \in \mathbb{C}^{2^n \times 2^n}$, if Assumptions \eqref{assumption: 1} and \eqref{assumption: 2} hold, then there exists a dictionary-based block encoding of $A$, which can be implemented with circuit depth $\mathcal{O}(\log (ns))$ and $\mathcal{O}(n^2s)$ ancillary qubits, using only single- and two-qubit gates.
\end{theorem}
\begin{proof}
    The dictionary-based block encoding of sparse matrices in Theorem \ref{theorem: sparse block encoding} comprises three oracles: $O_{c}$, PREP, and UNPREP. 
    \begin{itemize}
    \item The oracle $O_{c}$ in Equation \eqref{oracle Oc} can be implemented by five stages, which consists of 
    \begin{itemize}
        \item a Pauli-$X$ gate in Equation~\eqref{eq_O_c0};
        \item $(\lceil \log  s_0\rceil+n)$-index, $1$-word, $s$-sparse ranging SBM $O_{\rm c_1}$ in Equation~\eqref{eq_O_c1};
        \item $(\lceil \log  s_0\rceil+n+1)$-index, $n$-word, $s$-sparse mapping SBM $O_{\rm c_2}$ in Equation~\eqref{eq_O_c2};
        \item $(\lceil \log  s_0\rceil+n+1)$-index, $n$-word, $s$-sparse decomputing SBM $O_{\rm c_3}$ in Equation~\eqref{eq_O_c3};
        \item Controlled $2n$-qubit SWAP gate $O_{\rm SWAP}$ in Equation~\eqref{eq O_bit}.
    \end{itemize}
    
    By Lemma~\ref{SBM}, the oracles $O_{\rm c_1}$, $O_{\rm c_2}$  and $O_{\rm c_3}$ can be realized with \textit{circuit depth}
    \begin{equation*}
        \mathcal{O}\left(\log \left((\lceil \log  s_0\rceil+n+1)ns\right)\right)
    \end{equation*}
    and 
    \begin{equation*}
        \mathcal{O}\left((\lceil \log  s_0\rceil+n+1)ns\right)
    \end{equation*}
    ancillary qubits. Besides, the Pauli-$X$ gate and the controlled $2n$-qubit SWAP gate can be implemented with \textit{circuit depth} $\mathcal{O}(1)$ without ancillary qubits. 
		
    \item The oracles PREP in Equation \eqref{oracle: PREP} and UNPREP in Equation \eqref{oracle: UNPREP} prepare two $\lceil\log_{2}s_0\rceil$-qubit quantum states. By Lemma \ref{State Preparation}, these two oracles can be realized with \textit{circuit depth} $\mathcal{O}\left(\lceil\log_{2}s_0\rceil\right)$ and $\mathcal{O}\left(2^{\lceil\log_{2}s_0\rceil}\right)$ ancillary qubits.
\end{itemize}	
    Therefore, the dictionary-based block encoding of sparse matrices can be implemented with \textit{circuit depth}  
    \begin{equation*}
        \begin{aligned}
            \mathcal{O}\left(\log \left((\lceil \log  s_0\rceil+n+1)ns\right)\right) + \mathcal{O}\left(\lceil\log_{2}s_0\rceil\right) = \mathcal{O}\left(\log_{2}\left(ns\right)\right)
        \end{aligned}
    \end{equation*}
    and 
    \begin{equation*}
        \begin{aligned}
        &  \mathcal{O}\left((\lceil \log  s_0\rceil+n+1)ns\right) + \mathcal{O}\left(2^{\lceil\log_{2}s_0\rceil}\right)  = \mathcal{O}\left(n^2s \right)
        \end{aligned}
    \end{equation*}
    ancillary qubits.  
\end{proof}

\begin{remark}
    The task of block-encoding a $2^n \times 2^n$ sparse matrix with $s$ nonzero elements is analogous to prepare a $2n$-qubit $s$-sparse quantum state under the framework of third quantization~\cite{Prosen_2008}. While the optimal circuit depth for measurement-free sparse state preparation using single- and two-qubit gates is known to be $\Theta(\log(ns))$~\cite{zhang2022quantum}, the corresponding lower bound for measurement-free sparse matrix block encoding remains an open problem in quantum computation.
\end{remark}

In Appendix~\ref{section: Circuit Depth Comparison}, we analyze the \textit{circuit depth} of existing block-encoding protocols. Our dictionary-based framework outperforms these state-of-the-art methods, particularly for matrices with repeated elements, achieving low \textit{subnormalization} (Theorem~\ref{theorem: sparse block encoding}) and low \textit{circuit depth} (Theorem~\ref{theorem: circuit depth}). These advancements and their associated trade-offs are quantitatively summarized in Table~\ref{table: time metric}.

\section{Applications}
\label{section applications}

In this section, we demonstrate the practical utility of our dictionary data structure through two key applications: (i) signless Laplacian matrices of graphs and (ii) discrete differential operators. While the low-depth circuit implementation proposed in this work currently requires substantial ancillary qubit resources—rendering full-scale quantum circuit simulations impractical—we provide open-source Python code (\href{https://github.com/ChunlinYangHEU/DIBLE}{https://github.com/ChunlinYangHEU/DIBLE}) to validate the \textit{subnormalization} properties of our dictionary-based block encoding protocol of sparse matrices.

\label{section: applications}
\subsection{Signless Laplacian Matrix of Graph}
A weighted directed graph $G$ consists of a vertex set $V=\left[0,n-1\right]$, a directed edge set $E=\left\{e=\left(i,j\right): i,j\in V\right\}$ and a weight function $W=\left\{w(i):i\in V\right\}\cup\left\{w(i,j):\left(i,j\right)\in E\right\}$, where $w(i)$ is the weight of vertex $i\in V$ and $w(i,j)$ is the weight of directed edge $(i,j)\in E$. It is associated with a signless Laplacian matrix $Q_{G}$ \cite{cvetkovic2009introduction}, where $\left(Q_{G}\right)_{ij}=w(i,j)(i\neq j)$ and $\left(Q_{G}\right)_{ii}=w(i)$. 

Consider a weighted directed cyclic graph in Figure \ref{weighted directed cyclic graph} with vertex weights $w(i) = \alpha_1$, and edge weights $w(i,{\rm mod}(i-1,8)) = \alpha_2$ (clockwise), $w(i,{\rm mod}(i+1,8)) = \alpha_3$ (counterclockwise), where $\alpha_1,\alpha_2,\alpha_3\in\mathbb{R}$ are non-zero and $i\in [0,7]$.
\begin{figure}[h!]
    \centering
    \begin{minipage}{0.45\textwidth}
        \centering
        \begin{tikzpicture}
            [
            vertex/.style={circle, draw, fill=white, inner sep=2pt},
            edge/.style={-Latex, draw}
            ]
            \foreach \i in {0,...,7} {
                \node[vertex] (\i) at ({360/8*(\i)}:2cm) {\i};
            }
            \foreach \i in {0,...,7} {
                \pgfmathtruncatemacro{\next}{mod(\i+1,8)}
                \draw[edge] (\i) -- (\next);
                \draw[edge] (\next) -- (\i);
            }
        \end{tikzpicture}
    \end{minipage}
    \begin{minipage}{0.45\textwidth}
        \centering
        \begin{equation*}
            Q_{G} = \begin{bmatrix}
                \alpha_1 & \alpha_3 &  &  &  &  &  & \alpha_2 \\
                \alpha_2 & \alpha_1 & \alpha_3 &  &  &  &  &  \\
                 & \alpha_2 & \alpha_1 & \alpha_3 &  &  &  &  \\
                 &  & \alpha_2 & \alpha_1 & \alpha_3 &  &  &  \\
                 &  &  & \alpha_2 & \alpha_1 & \alpha_3 &  &  \\
                 &  &  &  & \alpha_2 & \alpha_1 & \alpha_3 &  \\
                 &  &  &  &  & \alpha_2 & \alpha_1 & \alpha_3 \\
                \alpha_3 &  &  &  &  &  & \alpha_2 & \alpha_1 \\
            \end{bmatrix}
        \end{equation*}
    \end{minipage}
    \caption{A weighted directed cyclic graph $G$ and its signless Laplacian matrix $Q_{G}$.}
    \label{weighted directed cyclic graph}
\end{figure}

For the signless Laplacian matrix $Q_{G}$, its dictionary data structure is expressed in Table~\ref{table: weighted adjacency matrix}.
\begin{table}[htbp]
    \centering
    \begin{tabular}{|c|l|}
        \hline
        \textbf{Keys} & \multicolumn{1}{c|}{\textbf{Values}} \\
        \hline
        $0$ & $\left\{(a_{ij},i,j):a_{ij}=\alpha_1,i=j,j\in[0,7]\right\}$ \\
        \hline 
        $1$ & $\left\{(a_{ij},i,j):a_{ij}=\alpha_2,(i,j)\in\left({\rm mod}(j+1,8),[0,7]\right)\right\}$ \\
        \hline 
        $2$ & $\left\{(a_{ij},i,j):a_{ij}=\alpha_3,(i,j)\in\left({\rm mod}(j-1, 8),[0,7]\right)\right\}$ \\
        \hline
    \end{tabular}
    \caption{The dictionary data structure of signless Laplacian matrix $Q_{G}$.}
    \label{table: weighted adjacency matrix}
\end{table}

There are $3$ data items in that dictionary (Table~\ref{table: weighted adjacency matrix}), and the number of non-zero elements is proportional to the dimension of $Q_G$, that is, $s_0=3$ and $s= 3\times2^n = 24$. The \textit{subnormalization} of the dictionary-based block encoding is the sum of modules of data values, that is, $\alpha = \left|\alpha_1\right| + \left|\alpha_2\right| + \left|\alpha_3\right|$.

Specifically, if we consider a weighted $k$-regular undirected graph $G$ with $n$ vertices, where for each vertex, it has a weight of $w_0$ and its $k$ adjacent edges have weights of $w_1, \dots, w_k$, respectively. The corresponding signless Laplacian matrix $Q_G$ has exactly $k+1$ non-zero elements in each row and column, $\alpha_0, \alpha_1, \dots, \alpha_k$, where $\alpha_0$ lies on the diagonal. This signless Laplacian can be encoded in an $n$-qubit system using a dictionary with parameters $s_0 = k+1$ and $s = (k+1)2^n$. The \textit{subnormalization} for the dictionary-based sparse block-encoding is given by $\sum_{i=0}^k |\alpha_i|$. For $n\ge\log (k+1)$, that \textit{subnormalization} is less than the Frobenius norm $\Vert Q_{G}\Vert_F$ by the Cauchy-Schwarz inequality as
$$
\Vert Q_{G}\Vert_F=\sqrt{2^n\sum_{i=1}^{k+1}\vert \alpha_i\vert^2}\ge\sqrt{(k+1)\sum_{i=1}^{k+1}\vert \alpha_i\vert^2 }\ge   \sum_{i=1}^{k+1}\left|\alpha_i\right| = \alpha.
$$

\subsection{Two-Dimensional Discrete Laplacian}
Consider the finite difference two-dimensional Laplacian which has been discussed in~\cite{sunderhauf2024block}. It is discretized by a regular grid of size $N_x\times N_y$ as
\begin{equation*}
    \Delta f(x_a, y_b) = \frac{f(x_{a-1},y_b) - 2 f(x_a, y_b) + f(x_{a+1}, y_b)}{(\Delta x)^2} + \frac{f(x_a,y_{b-1}) - 2 f(x_a, y_b) + f(x_{a}, y_{b+1})}{(\Delta y)^2}.
\end{equation*}
To encode the values of the grid, the grid is reshaped as an $N_xN_y$-dimensional vector,
\begin{equation*}
    f_{a+bN_x} = f(x_a,y_b),\ a\in[0, N_x-1],\ b\in[0,N_y-1].
    \label{eq:2d encoding}
\end{equation*}
Then, all non-zero elements of Laplacian matrix $A$ are given by
\begin{equation*}
    A_{a_1+b_1N_x, a_2+b_2N_x} = \begin{cases} A_0 := -2(1/(\Delta x)^2 + 1/(\Delta y)^2), & \text{for}\ a_1=a_2,\ b_1=b_2, \\
    A_1 := 1/(\Delta x)^2, & \text{for}\ |a_1-a_2|=1,\  b_1=b_2, \\
    A_2 := 1/(\Delta y)^2,  & \text{for}\ |b_1-b_2|=1,\ a_1=a_2. \\
    \end{cases}
\end{equation*}
For $N_x=4$, $N_y=4$, the matrix is in form of 
\begin{equation}
\begin{tikzpicture}
  \matrix (m)[
    matrix of math nodes,
    nodes in empty cells,
    left delimiter=(,
    right delimiter=)
  ] {
       A_0    & A_1  &  &  &  A_2     & &    \\
       A_1 & A_0   & A_1 & &  & A_2    & &  \\
       & A_1    &  A_0  & A_1  &   &    & A_2  \\
       & & A_1 & A_0 & & & & A_2 \\
       A_2&&&& A_0    & A_1  &&& A_2   \\
       &A_2&&& A_1 & A_0   & A_1 &&& A_2    \\
       &&A_2&& & A_1    &  A_0  & A_1&&&A_2  \\
       &&&A_2& & & A_1 & A_0  &&&& A_2 \\
       &&&&A_2&&&& A_0    & A_1  &  &  &  A_2     & &    \\
       &&&&&A_2&&&A_1 & A_0   & A_1 & &  & A_2    & &  \\
       &&&&&&A_2&&& A_1    &  A_0  & A_1  &   &    & A_2  \\
       &&&&&&&A_2&& & A_1 & A_0 & & & & A_2 \\
       &&&&&&&&A_2&&&& A_0    & A_1     \\
       &&&&&&&&&A_2&&&A_1 & A_0   & A_1  \\
       &&&&&&&&&&A_2&&& A_1    &  A_0  & A_1  \\
       &&&&&&&&&&&A_2&& & A_1 & A_0 \\
  };
  \draw[dashed] (m-1-1.north west) rectangle (m-4-4.south east);
  \draw[dashed] (m-5-5.north west) rectangle (m-8-8.south east);
  \draw[dashed] (m-9-9.north west) rectangle (m-12-12.south east);
  \draw[dashed] (m-13-13.north west) rectangle (m-16-16.south east);
\end{tikzpicture}
\label{eq:poisson structure}
\end{equation}
This matrix can be encoded using a dictionary as stated in Table~\ref{table: Laplacian}, where the encoding parameters $s_0=5$, $s=64$. The matrix can be encoded using the dictionary-based block encoding with \textit{subnormalization} $\alpha = |A_0| + 2(|A_1| + |A_2|)$.

\begin{table}[h!]
    \centering
    \begin{tabular}{|c|l|}
        \hline
        \textbf{Keys} & \multicolumn{1}{c|}{\textbf{Values}} \\
        \hline
        $0$ & $\left\{(a_{ij},i,j): a_{ij}=A_{0}, (i,j)\in\left(j,[0,15]\right\}\right)$ \\
        \hline
        $1$ & $\left\{(a_{ij},i,j): a_{ij}=A_{1}, (i,j)\in\left(j-1,\left\{j\in[0,15]:{\rm mod}(j,4)\neq0\right\}\right)\right\}$ \\
        \hline
        $2$ & $\left\{(a_{ij},i,j): a_{ij}=A_{1}, (i,j)\in\left(j+1,\left\{j\in[0,15]:{\rm mod}(j,4)\neq3\right\}\right)\right\}$ \\
        \hline
        $3$ & $\left\{(a_{ij},i,j): a_{ij}=A_{2}, (i,j)\in\left(j-4,[4,15]\right)\right\}$ \\
        \hline
        $4$ & $\left\{(a_{ij},i,j): a_{ij}=A_{2}, (i,j)\in\left(j+4,[0,11]\right)\right\}$ \\
        \hline
    \end{tabular}
    \caption{The dictionary data structure of two-dimensional discrete Laplacian $A$ in Equation \eqref{eq:poisson structure}.}
    \label{table: Laplacian}
\end{table}

\subsection{Matrices in Ocean Acoustic GEPs}
Ocean acoustics plays a crucial role in understanding underwater sound propagation. 
The problem of solving the acoustic field essentially reduces to solving the following differential equations, 
\begin{equation*}
	\frac{d^2 \varphi_{m}\left(z\right)}{dz^2} + \left(\frac{\omega^{2}}{c^{2}\left(z\right)} - k_{m}^{2}\right)\varphi_{m}\left(z\right) = \boldsymbol{0},
\end{equation*}
\begin{equation*}
	\begin{cases}
		\frac{d\xi_{m,1}\left(z\right)}{dz} = -\xi_{m,2}\left(z\right) + \frac{1}{\mu(z)} \xi_{m,4}\left(z\right), \\
		\frac{d\xi_{m,2}\left(z\right)}{dz} = \frac{\lambda(z)k_m^2}{\lambda(z)+2\mu(z)} \xi_{m,1}\left(z\right) + \frac{1}{\lambda(z)+2\mu(z)} \xi_{m,4}\left(z\right), \\
		\frac{d\xi_{m,3}\left(z\right)}{dz} = \left(\frac{4\mu(z)\left(\lambda(z)+\mu(z)\right)k_m^2}{\lambda(z)+2\mu(z)} - \rho(z)\omega^2\right) \xi_{m,1}\left(z\right) - \frac{\lambda(z)}{\lambda(z)+2\mu(z)}\xi_{m,4}\left(z\right), \\
		\frac{d\xi_{m,4}\left(z\right)}{dz} = -\rho(z)\omega^2 \xi_{m,2}\left(z\right) + k_m^2 \xi_{m,3}\left(z\right).
	\end{cases}
\end{equation*}
which respectively describe the propagation of sound pressure modes $\varphi_m(z)$ in seawater and the propagation of stress-displacement modes $\xi_m(z)=(\xi_{m,1}(z),\xi_{m,2}(z),\xi_{m,3}(z),\xi_{m,4}(z))^{\rm T}$ in ice. The quantities to be determined include the sound pressure modes, the stress-displacement modes, and the corresponding wave numbers $k_m$.

Using the finite difference method, it can be further transformed into the following generalized eigenvalue problem~\cite{aki2002quantitative,jensen2011computational}, 
\begin{equation}\label{equation: GEPs}
    AV=BV\Sigma.
\end{equation}
The details of this problem and the general structure of $A$ and $B$ are shown in Appendix \ref{section: GEPs in ocean acoustics}. The dimensions of the matrices $A$ and $B$ are $4N_1+N_2+5$, where $N_1$ and $N_2$ represent the numbers of stratified grid layers for ice and seawater, respectively. These two matrices can be encoded in a system with $n=\lceil \log (4N_1+N_2+5)\rceil$ qubits. Based on their structure, their dictionary data structures are constructed and shown in Tables~\ref{table: dictionary of A} and \ref{table: dictionary of B} of Appendix~\ref{section: GEPs in ocean acoustics}, where the dictionaries satisfy that the data functions $c_l(j)$ are all injective.

For matrix $A$ in Equation~\eqref{equation: GEPs}, there are $21$ data items in the encoding dictionary of $A$ and the non-zero count is associated with $N_1$ and $N_2$. Thus, the encoding parameters of $A$ are $s_0=19$ and $s=5+24N_1+3N_2$. The sum of the modules of all data values serves as the \textit{subnormalization} in its block encoding, that is,
$$\left|a_{0}\right| + 2(\left|a_{1}\right| + \left|a_{2}\right| + \left|a_{3}\right| + \left|a_{4}\right| + \left|a_{5}\right| + \left|a_{6}\right|) + \left|a_{7}\right| + \left|a_{8}\right| + \left|a_{9}\right| + \left|a_{10}\right| + \left|a_{11}\right| + \left|a_{12}\right|.$$

For the matrix $B$ in Equation~\eqref{equation: GEPs}, the number of data items is $s_0=9$, and the non-zero elements count is $s=3+6N_1+N_2$. Besides, its \textit{subnormalization} of the block encoding derived from the associated dictionary is
$$
2(\left|b_0\right|+\left|b_1\right|+\left|b_2\right|)+\left|b_3\right|+\left|b_4\right|+\left|b_5\right|.
$$

\section{Conclusion and Outlook}
\label{section: conclusion and outlook}

In this work, we have advanced the implementation of sparse matrix block encodings through optimized \textit{circuit depth} architectures. 

We have proposed a unified dictionary data structure (Table~\ref{dictionary}) that generalizes multiple block-encoding protocols of sparse matrices~\cite{gilyen2019quantum,parker2024sfable,camps2024explicit,sunderhauf2024block}. Our proposed dictionary-based block-encoding framework of sparse matrices was constructed using the dictionary data function specified in Table~\ref{table: dictionary data structure in this paper}. The dictionary-based block encoding for complex matrices (Theorem~\ref{theorem: sparse block encoding}) and its Hermitian extension for non-negative symmetric matrices (Theorem~\ref{theorem: sparse Hermitian block encoding}) generalized the results in \cite{camps2024explicit}, achieving improved \textit{subnormalization}. Furthermore, we have demonstrated that the dictionary-based block encoding of sparse matrices can be viewed as a special case of linear combination of unitaries.

Based on the sparse Boolean memory architecture, we have established fundamental bounds on quantum \textit{circuit depth} requirements for the implementation of the dictionary-based block encoding of sparse matrices (Theorems~\ref{theorem: circuit depth}). Our theoretical analysis revealed a doubly logarithmic scaling relationship: the required \textit{circuit depth} demonstrates $\mathcal{O}(\log n)$ dependence on the Hilbert space dimension $N=2^n$, while maintaining $\mathcal{O}(\log s)$ scaling with respect to the number of non-zero elements $s$. To facilitate the rigorous comparison across different block-encoding protocols, we propose a novel time-efficiency metric that combines \textit{circuit depth} with \textit{subnormalization}. For matrices with repeated elements, our approach achieved superior \textit{subnormalization} and \textit{circuit depth} compared to existing sparse block-encoding protocols (Table~\ref{table: time metric}). We have demonstrated the effectiveness of our method through several applications, including graph problems, two-dimensional discrete Laplacian operators, and ocean acoustic GEPs.

While our block-encoding scheme achieved an efficient \textit{time metric}, two fundamental questions remain open: (1) the optimality of our \textit{circuit depth} $\mathcal{O}(\log (ns))$, and (2) the gap between our \textit{subnormalization} $\alpha = \sum_{l}|A_{l}|$ and the theoretical minimum $\|A\|_{\rm sp}$. These observations reveal promising directions for future optimization.

\section*{Acknowledgments}
This work is supported by the Stable Supporting Fund of Acoustic Science and Technology Laboratory (Grant No. JCKYS2024604SSJS001), the Fundamental Research Funds for the Central Universities (Grant No. 3072024XX2401), the Hong Kong Research Grants Council (RGC) (Grant No. 15213924), and the National Natural Science Foundation of China (Grant No. 62173288).

\bibliography{references}

\appendix
\section{Block Encoding and Dictionary Data Structure}\label{section: details of block encoding}
Block encoding is a technique that represents a non-unitary operator $A$ as a subsystem of a larger unitary operator $U_A$. 
\begin{definition}[Block Encoding \cite{gilyen2019quantum}]
	Suppose that $A$ is an $n$-qubit operator, $N=2^{n},\alpha,\epsilon\in\mathbb{R}_{+}$ and $a \in \mathbb{N}$, then we say that the $(n + a)$-qubit unitary $U_{A}$ is an $(\alpha,a,\epsilon )$ block encoding of A, if
	\begin{equation*}
		\left\|A-\alpha\left(\bra{0}^{\otimes a}\otimes I_{N} \right) U_{A} \left(\ket{0}^{\otimes a} \otimes I_{N} \right)   \right\| \leq \epsilon.
	\end{equation*}
\end{definition}

A Hermitian block encoding is a specialized form of block encoding whose $U_{A}$ is not only unitary but also Hermitian. Consequently, the encoded matrix $A$ must also be Hermitian.
\begin{definition}[Hermitian Block Encoding(\cite{linlin2022})]
	Let $U_{A}$ be an $\left(\alpha, a, \epsilon\right)$ block encoding of a Hermitian matrix $A$. If $U_{A}$ is also Hermitian, then it is called an $\left(\alpha, a, \epsilon\right)$ Hermitian block encoding of $A$. When $\epsilon= 0$, it is called an $\left(\alpha, a\right)$ Hermitian block encoding. The set of all $\left(\alpha, a, \epsilon\right)$ Hermitian block encodings of $A$ is denoted by ${\rm HBE}_{\alpha, a}\left(A,\epsilon\right)$, and ${\rm HBE}_{\alpha, a}\left(A\right)={\rm HBE}_{\alpha, a}\left(A,0\right)$. 
\end{definition}

\subsection{Several Block Encoding Protocols of Sparse Matrices within the Framework of Dictionary Data Structure}\label{subsection: several sparse block encodings within dictionary}

Different dictionary data structures result in different block-encoding protocols. Within the unified framework proposed in section~\ref{subsection: Dictionary Data Structure}, we can summarize several block-encoding protocols of sparse matrices in~\cite{gilyen2019quantum,parker2024sfable,camps2024explicit,sunderhauf2024block}. Given triplets $(a_{ij}, i, j)$ that encode a sparse matrix, these protocols satisfy the following row and column rules, respectively.

\begin{enumerate}
    \item A sparse matrix can be characterized by two data functions: $j = r(i,k)$ locating the $k$-th non-zero element in the $i$-th row ($k \in [0,S_r-1]$) and $i = c(l,j)$ finding the $l$-th non-zero element in the $j$-th column ($l \in [0,S_c-1]$), where $S_r$/$S_c$ denotes the row/column sparsities. Based on this, the triplets $(a_{ij},i,j)$ form the dictionary data structure of~\cite{gilyen2019quantum}.
    
    \item The positions of elements in a sparse matrix can be fully characterized by a vectorized index mapping $\alpha = iN + j$, where $N$ is the dimension of the matrix. All triplets $(a_{ij},i,j)$ with the vectorized indices $\alpha$ form the fundamental data items of the dictionary data structure of~\cite{parker2024sfable}, enabling efficient storage and retrieval.

    \item The positions of non-zero elements in a sparse matrix can be characterized by a bijective function $i = c(l,j)$ mapping $(l,j)$ pairs to row indices $i$, where $j$ denotes the column index and $l$ enumerates non-zero elements in every column. All non-zero elements labeled by $l$ should share both the same value $A_l$ and the same function $c(l,j)$. The corresponding triplets form a data item, which constitutes the dictionary data structure of~\cite{camps2024explicit}. It compactly represents the sparse matrix through its non-zero pattern and value distribution.

    \item The row and column index sets of a sparse matrix are characterized by two data functions: $(j, s_c) = c(d, m)$ for column indices and column sparsity indices, and $(i, s_r) = r(d, m)$ for row indices and row sparsity indices, where $d \in[0,D-1]$ is the data index, $D$ is the number of distinct non-zero values $A_d$, $m \in [0,M-1]$ is the multiplicity index, $M$ is the maximum multiplicity of any of $D$ data items. All non-zero elements sharing the same value $A_d$ under identical functions $c(d, m)$ and $r(d, m)$ form equivalence classes that constitute the dictionary data structure of~\cite{sunderhauf2024block}, providing an efficient representation of sparse matrices through their value distributions and index patterns.
\end{enumerate}

\section{Circuit Depth Comparison}
\label{section: Circuit Depth Comparison}
To demonstrate the circuit-depth efficiency of our block-encoding framework, we conduct a rigorous comparative analysis in this section. Specifically:
\begin{itemize}
    \item In Lemma~\ref{lemma: circuit depth of Sunderhauf}, we analyze the low \textit{circuit depth} implementation of PREP/UNPREP block encoding for sparse matrices~\cite{sunderhauf2024block}, which currently achieves the lowest known \textit{subnormalization} factor for sparse matrices. 
    \item In Lemma~\ref{lemma: comparison depth 2}, we analyze the sparse implementation for the low-circuit-depth block-encoding protocol~\cite{clader2022quantum,zhang2024circuit}, which currently represents the most depth-efficient approach for general matrices.
\end{itemize}
These carefully selected benchmarks provide a comprehensive basis for evaluating our improvements in circuit-depth complexity.

\begin{lemma}
\label{lemma: circuit depth of Sunderhauf}
    Given a matrix $A\in\mathbb{C}^{2^{n}\times 2^{n}}$ with $D$ distinct non-zero elements, row sparsity $S_r$ and column sparsity $S_c$. If $\lceil \log D\rceil \leq \lceil \log S_c\rceil, \lceil \log S_r\rceil$, then the {\rm PREP/UNPREP} block encoding in \cite{sunderhauf2024block} can be implemented with circuit depth $\mathcal{O}\left(n2^{n/2}\right)$ and $n_{\rm anc}$ ancillary qubits using only single- and two-qubit gates, where $n_{\rm anc}\geq \Omega\left(4^{n}/n\right)$.
\end{lemma}

The proof is given in Appendix~\ref{subsection_proof_lemma}.

\begin{lemma}\label{lemma: comparison depth 2}
    Given a matrix  $A\in\mathbb{C}^{2^{n}\times 2^{n}}$ with $s$ non-zero elements, where $\lceil\log s\rceil = \mathcal{O}(n)$. Then there exists a controlled state preparation oracle $U_L$ and a state preparation oracle $U_R$ such that
    \begin{equation}
        \begin{aligned}
        U_R\ket{0}^{\otimes n}\ket{j} = \ket{A_j}\ket{j},\quad 
        U_L\ket{0}^{\otimes n} =\ket{A},
        \end{aligned}
        \label{eq_UR_UL}
    \end{equation}
    where $\ket{A_j} = \sum_{k=0}^{2^n-1}\frac{a_{jk}}{\Vert A_{j,\cdot}\Vert_F} \ket{k}$, $\ket{A} = \sum_{j=0}\frac{\Vert A_{j,\cdot}\Vert_F}{\Vert A\Vert_F}\ket{j}$, $A_{j,\cdot}$ denotes the $j$-th row of $A$. Then $U_A= U_R^\dagger\text{SWAP}(U_L\otimes I_{2^n})$ is a block-encoding of $A$ with subnormalization  $\Vert A\Vert_F$, as stated in~\cite{clader2022quantum}, and it can be implemented with circuit depth $\mathcal{O}\left(n\right)$ and $\mathcal{O}\left(2^{2n}\right)$ ancillary qubits, using only single- and two-qubit gates. 
\end{lemma}

The \textit{circuit depth} in these two protocols~\cite{sunderhauf2024block,clader2022quantum} (Lemma~\ref{lemma: circuit depth of Sunderhauf} and Lemma~\ref{lemma: comparison depth 2}) primarily arises from their dependence on two computationally costly operations: unitary synthesis procedures and controlled state preparation routines~\cite{Yuan2023optimalcontrolled}. Both approaches inherently require considerable \textit{circuit depth}, a factor that proves especially critical in the context of sparse matrix encodings.

\subsection{Proof of Lemma \ref{lemma: circuit depth of Sunderhauf}}
\label{subsection_proof_lemma}
As stated in Appendix~\ref{subsection: several sparse block encodings within dictionary}, the PREP/UNPREP block encoding scheme in~\cite{sunderhauf2024block} can be described by a dictionary data structure, which is constructed in Table~\ref{table: dictionary data structure in sunderhauf2024block}. The PREP/UNPREP block encoding scheme relies on the equality $\lceil \log_{2}D \rceil + \lceil \log_{2}M \rceil = \lceil \log_{2}N \rceil + \lceil \log_{2}S \rceil$ and the assumption $\lceil \log_{2}D \rceil \leq \lceil \log_{2}S_c \rceil,\lceil \log_{2}S_r \rceil$, where $D$ represents the number of distinct data values, $M$ denotes the maximum multiplicity of data values, $S = \max\{S_c,S_r\}$ captures the matrix sparsity (with $S_c$ and $S_r$ being column and row sparsities, respectively), and $N$ is the dimension of the matrix.
\begin{table}[h!]
    \centering
    \begin{tabular}{|c|c|}
            \hline
            \textbf{Keys} & \textbf{Values} \\
            \hline
            0 & $\left\{(a_{i,j},i,j) : a_{ij} = A_0, (i,j)\in \left(S_r(0), S_c(0)\right), (i,s_r)=r(0,m),(j,s_c)=c(0,m) \right\}$ \\
            \hline
            1 & $\left\{(a_{i,j},i,j) : a_{ij} = A_1, (i,j)\in \left(S_r(1), S_c(1)\right), (i,s_r)=r(1,m),(j,s_c)=c(1,m)  \right\}$ \\
            \hline
            2 & $\left\{(a_{i,j},i,j) : a_{ij} = A_2, (i,j)\in \left(S_r(2), S_c(2)\right), (i,s_r)=r(2,m),(j,s_c)=c(2,m)  \right\}$ \\
            \hline
            \vdots & \vdots \\
            \hline
        \end{tabular}
        \caption{A dictionary data structure for implementing the PREP/UNPREP block-encoding scheme in~\cite{sunderhauf2024block}.}
        \label{table: dictionary data structure in sunderhauf2024block}
\end{table}

\begin{figure}[h!]
    \centering
    \begin{tikzpicture}
        \begin{yquant}
            qubit del;
            qubit d;
            qubit {s} f2;
            qubit {$\ket{j}$} block;
            
            ["north:$\lceil\log D\rceil$" {font=\protect\footnotesize, inner sep=0pt}]
            slash d;
            ["north:$\lceil\log S\rceil-\lceil\log D\rceil$" {font=\protect\footnotesize, inner sep=0pt}]
            slash f2;
            ["north east:$n$" {font=\protect\footnotesize, inner sep=0pt}]
            slash block;
            hspace {5pt} -;
            box {$H_{S/D}$} f2;
            box {PREP} d;
            
            align -;
            [shape=yquant-init, decoration={mirror}] 
            inspect {$s_c$} (d,f2);
            text {$j$} block;
            box {$O_c^\dag$} (f2, block, d);
            discard f2;
            text {$d$} d;
            text {$m$} block;
            ["north:$\lceil\log D\rceil$" {font=\protect\footnotesize, inner sep=0pt}]
            slash d;
            ["north:$\lceil\log M\rceil$" {font=\protect\footnotesize, inner sep=0pt}]
            slash block;
        
            [plusctrl, shape=yquant-circle]
            box {$O_\text{rg}$} (block, f2,d) | del;

            text {$d$} d;
            text {$m$} block;
            hspace {5pt} -;
            box {$O_r$} (f2, block,d);
            settype {qubit} f2;
            inspect {$s_r$} (f2,d);
            text {$i$} block;
            hspace {5pt} -;
            box {UNPREP} d;
            box {$H_{S/D}^\dagger$} f2;
            hspace {5pt} -;
            output {$\ket{i}$} block;
        \end{yquant}
    \end{tikzpicture}
    \caption{Basic framework of PREP/UNPREP block encoding in~\cite{sunderhauf2024block} with dictionary data structure in Table~\ref{table: dictionary data structure in sunderhauf2024block}, where $H_{S/D} = H^{\otimes(\lceil\log S\rceil-\lceil\log D\rceil)}$.}
    \label{circuit: PREP/UNPREP sunderhauf}
\end{figure}

These ensure the versatility in their implementation with the quantum circuit in Figure~\ref{circuit: PREP/UNPREP sunderhauf}. The oracles $O_{c}$ and $O_{r}$ are used to implement the mapping of data functions $c(d,m)$ and $r(d,m)$, respectively, while the oracle $O_{\rm rg}$ identifies the defining domain of these two functions. The oracles PREP and UNPREP are used to prepare the states encoding all distinct data values. The specific statements are as follows.
\begin{itemize}
    \item The out-of-range oracle $\tilde{O}_{\rm rg}$ satisfies
    \begin{equation}
        \tilde{O}_{\rm rg}\ket{d}\ket{m}\ket{0}_{\rm del} = \begin{cases}
            \ket{d}\ket{m}\ket{0}_{\rm del}, & \mbox{if } A_{i(d,m),j(d,m)} = A_d, \\
            \ket{d}\ket{m}\ket{1}_{\rm del}, & \mbox{if } A_{i(d,m),j(d,m)} = 0,
        \end{cases}
        \label{eq_tilde_O_rg}
    \end{equation}
    where $d\in\left[0,D-1\right]$, $m\in\left[0,M-1\right]$.
    
    \item The column and row oracles $\tilde{O}_{c}, \tilde{O}_{r}$ satisfy
    \begin{equation}
        \tilde{O}_{\rm c}\ket{d}\ket{m} = \ket{j}\ket{s_c},\quad
        \tilde{O}_{\rm r}\ket{d}\ket{m} = \ket{i}\ket{s_r}
        \label{eq_tilde_Oc_Or},
    \end{equation}
    where $d\in\left[0,D-1\right]$, $m\in\left[0,M-1\right]$, $i,j\in\left[0,N-1\right]$, $s_c\in\left[0,S_c-1\right]$, $s_r\in\left[0,S_r-1\right]$.
    
    \item The state preparation oracles PREP and UNPREP satisfy
    \begin{gather*}
        \mbox{PREP}\ket{0}^{\otimes \lceil\log_{2}D\rceil} = \frac{1}{\sqrt{\sum_{d=0}^{D-1}\left|A_d\right|}}\sum_{d=0}^{D-1}\mbox{sgn}(A_{d})\sqrt{\left|A_d\right|}\ket{d}, \\
        \bra{0}^{\otimes \lceil\log_{2}D\rceil}\mbox{UNPREP} = \frac{1}{\sqrt{\sum_{d=0}^{D-1}\left|A_d\right|}}\sum_{d=0}^{D-1}\sqrt{\left|A_d\right|}\bra{d}.
    \end{gather*}
\end{itemize}

Noting that the column and row oracles $\tilde{O}_{c},\tilde{O}_{r}$  are unitary synthesis, which takes relatively high \textit{circuit depth} using the state-of-the-art techniques~\cite{Yuan2023optimalcontrolled}, compared to the column oracles $O_c$ \eqref{oracle Oc} in this article. 

\begin{lemma}[Unitary Synthesis~\cite{Yuan2023optimalcontrolled}]\label{Unitary synthesis}
    For any $m \geq 0$, any $n$-qubit unitary can be implemented by a quantum circuit with $m$ ancillary qubits, using single-qubit gates and CNOT gates, of depth $\mathcal{O}\left(\frac{n^{1/2}2^{3n/2}}{m^{1/2}}\right)$ when $\Omega\left(2^n/2\right) \leq \mathcal{O}\left(4^n/n\right)$. In particular, the depth is $\mathcal{O}\left(n2^{n/2}\right)$ when $m\geq \Omega\left(4^n/n\right)$.
\end{lemma}

\begin{proof}[Proof of Lemma \ref{lemma: circuit depth of Sunderhauf}]
The PREP/UNPREP block encoding scheme in \cite{sunderhauf2024block} comprises five oracles: $\tilde{O}_{\rm rg}$, $\tilde{O}_{c}$, $\tilde{O}_{r}$, PREP, and UNPREP.
\begin{itemize}
    \item The oracle $\tilde{O}_{\rm rg}$ can be implemented by SBM. We introduce a $\left(\lceil\log_{2}D\rceil+\lceil\log_{2}M\rceil\right)$-index, $1$-word, $\tilde{s}_{\rm rg}$-sparse Boolean function $\tilde{f}_{\rm rg}$ such that 
     \begin{equation*}
        \tilde{f}_{\rm rg}(d,m) = \begin{cases}
            0 & \mbox{if } A_{i(d,m),j(d,m)} = A_d, \\
            1 & \mbox{if } A_{i(d,m),j(d,m)} = 0,
        \end{cases}
    \end{equation*}
    where the sparsity is given by
    \begin{equation*}
        \tilde{s}_{\rm rg} = \left|\left\{(d,m):\tilde{f}_{\rm rg}(d,m)=1\right\}\right| = \left|\left\{(d,m):A_{i(d,m),j(d,m)}=0\right\}\right| = \tilde{s}_1 < 4^n,
    \end{equation*}
    $\tilde{s}_1$ is the number of invalid $(d,m)$ and always less than the total number of elements of the matrix. Then, there exits a SBM such that
    \begin{equation*}
        \mbox{select}(\tilde{f}_{\rm rg})=\sum_{d,m=0}^{2^{\lceil\log_{2}D\rceil}-1, 2^{\lceil\log_{2}M\rceil}-1}\ket{d,m}\bra{d,m}\otimes\left(\tilde{f}_{\rm rg}(d,m)X+\overline{\tilde{f}_{\rm rg}}(d,m)I\right).
    \end{equation*}
    By Lemma \ref{SBM}, it can be implemented with \textit{circuit depth} of  $\mathcal{O}\left(\log_{2}\left(\left(\lceil\log_{2}D\rceil+\lceil\log_{2}M\rceil\right) \tilde{s}\right)\right)$ and ancillary qubits of $\mathcal{O}\left(\left(\lceil\log_{2}D\rceil+\lceil\log_{2}M\rceil\right) \tilde{s}\right)$. 

    \item The oracles $\tilde{O}_{c}$ and $\tilde{O}_{r}$ are two $\left(\lceil\log_{2}D\rceil+\lceil\log_{2}M\rceil\right)$-qubit unitary synthesis. By Lemma \ref{Unitary synthesis}, they can be implemented with \textit{circuit depth}  
    \begin{equation*}
        \mathcal{O}\left(\left(\lceil\log_{2}D\rceil+\lceil\log_{2}M\rceil\right)2^{\frac{\lceil\log_{2}D\rceil+\lceil\log_{2}M\rceil}{2}}\right) = \mathcal{O}\left(n2^{n/2}\right)
    \end{equation*}
    and $n_{\rm anc}$ ancillary qubits, where
    \begin{equation*}
        n_{\rm anc} \geq \Omega\left(4^{\lceil\log_{2}D\rceil+\lceil\log_{2}M\rceil}/\left(\lceil\log_{2}D\rceil+\lceil\log_{2}M\rceil\right)\right) = \Omega\left(4^n/n\right).
    \end{equation*}
        
    \item The oracles PREP and UNPREP prepare two $\lceil log_{2}D \rceil$-qubit states. Using the state preparation in Lemma \ref{State Preparation}, they can be implemented with \textit{circuit depth} $\mathcal{O}\left(\lceil log_{2}D \rceil\right)$ and $\mathcal{O}\left(2^{\lceil log_{2}D \rceil}\right)$ ancillary qubits using only single- and two-qubit gates.
\end{itemize}
Since $\lceil \log_{2}D \rceil + \lceil \log_{2}M \rceil =\lceil \log_{2}N \rceil + \lceil \log_{2}S \rceil = \mathcal{O}(n)$, we have $\lceil \log_{2}D \rceil=\mathcal{O}\left(n\right)$ and $\lceil \log_{2}M \rceil=\mathcal{O}\left(n\right)$. Therefore, its PREP/UNPREP block encoding can be implemented with \textit{circuit depth}
\begin{equation*}
    \begin{aligned}
    \mathcal{O}\left(\log_{2}\left(\left(\lceil\log_{2}D\rceil+\lceil\log_{2}M\rceil\right) \tilde{s}_1\right)\right) + \mathcal{O}\left(n2^{n/2}\right) + \mathcal{O}\left(\lceil log_{2}D \rceil\right) = \mathcal{O}\left(n2^{n/2}\right)
    \end{aligned}
\end{equation*}
and 
\begin{equation*}
    \mathcal{O}\left(\left(\lceil\log_{2}D\rceil+\lceil\log_{2}M\rceil\right) \tilde{s}_1\right) + n_{\rm anc} + \mathcal{O}\left(2^{\lceil log_{2}D \rceil}\right) = n_{\rm anc} \geq \Omega\left(4^{n}/n\right)
\end{equation*}
ancillary qubits using only single- and two-qubit gates.
\end{proof}

\subsection{Proof of Lemma \ref{lemma: comparison depth 2}}

\begin{lemma}
    [Controlled Quantum State Preparation~\cite{Yuan2023optimalcontrolled}]
    For any integers $k,m>0$, $n>0$ and any quantum state $\left\{\ket{\psi_j}:j\in\{0,1\}^k\right\}$, the following controlled quantum state preparation
    $$
    \ket{j}\ket{0}^{\otimes n}\to \ket{j}\ket{\psi_j},\ \forall j\in\{0,1\}^{k}
    $$
    can be implemented by a quantum circuit consisting of single-qubit and CNOT gates of depth $\mathcal{O}\left(n+k+\frac{2^{n+k}}{n+k+m}\right)$ and size $\mathcal{O}\left(2^{n+k}\right)$ with ancillary qubits. These bounds are optimal for any $k,m>0$.
    \label{lemma: Controlled Quantum State Preparation}
\end{lemma}

\begin{proof}[Proof of Lemma \ref{lemma: comparison depth 2}]
    We only need to demonstrate the \textit{circuit depth}. This block encoding comprises three oracles: $U_R$, $U_L$, and $\text{SWAP}$.
    \begin{itemize}
        \item $U_R$ is can represented as 
        $$
        \ket{0}^{\otimes n}\ket{j}\xrightarrow{U_R} \ket{A_j}\ket{j}, \ j\in\{0,1\}^{n}
        $$
        which can be implemented by controlled-state preparation with \textit{circuit depth} $\mathcal{O}\left(n\right)$ and $\mathcal{O}\left(2^{2n}\right)$ ancillary qubits by Lemma \ref{lemma: Controlled Quantum State Preparation};
        \item $U_L$ is an $n$-qubit $s$-sparse state-preparation, which can be implemented with \textit{circuit depth} $\mathcal{O}(\log (ns))$ and $\mathcal{O}(ns\log  s)$ ancillary qubits by Lemma~\ref{Sparse State Preparation};
        \item $2n$-qubit SWAP gate can be implemented with \textit{circuit depth} $\mathcal{O}(1)$ without ancillary qubits.
    \end{itemize}
    Above all, $U_A$ can be implemented with \textit{circuit depth}  
    \begin{equation*}
        \mathcal{O}\left(n\right) + \Theta(\log (ns))  = \mathcal{O}\left(n\right)
    \end{equation*}
    and 
    \begin{equation*}
        \mathcal{O}\left(2^{2n}\right) + \mathcal{O}(ns\log  s)) = \mathcal{O}\left(2^{2n}\right)
    \end{equation*}
    ancillary qubits.  
\end{proof}

\section{Generalized Eigenvalue Problems in Ocean Acoustic Modeling}
\label{section: GEPs in ocean acoustics}
In the field of ocean acoustics, the propagation of acoustic waves through seawater and ice is a complex process crucial for obtaining underwater information. This information is essential for various applications, including underwater communication, navigation, and environmental monitoring. Therefore, the study of underwater acoustic propagation is a core component of research in underwater acoustic information. 

Ocean acoustics encompasses several subfields, including shallow sea acoustics, deep sea acoustics, and polar acoustics. In polar acoustics, the interaction of acoustic waves with ice and seawater presents distinct challenges. Typically, sound pressure fields are measured in seawater, whereas displacement fields are measured in ice  \cite{jensen2011computational, aki2002quantitative}. 

Utilizing the normal-mode method, on one hand, the sound pressure field $p\left(r,z\right)$ can be computed from the sound pressure modes $\varphi_{m}\left(z\right)$ \cite{jensen2011computational},
\begin{equation*}
	p\left(r,z\right) = \frac{\imath}{\rho(z_{s}) \sqrt{8\pi r}} e^{-i\pi /4} \sum_{m}\varphi_{m}\left(z_{s}\right) \varphi_{m}\left(z\right) \frac{e^{ik_{m}r}}{\sqrt{k_{m}}},
\end{equation*}
where $r$ is the horizontal distance, $z$ is the depth, $z_{s}$ is the depth of the sound source, $\rho(z_{s})$ is the density of the medium at $z_{s}$, $k_{m}$ is the horizontal wave number and $\imath^{2}=-1$. Meanwhile, the sound pressure modes $\varphi_{m}\left(z\right)$ satisfy the Helmholtz equation
\begin{equation}\label{equation: seawater}
	\frac{d^2 \varphi_{m}\left(z\right)}{dz^2} + \left(\frac{\omega^{2}}{c^{2}\left(z\right)} - k_{m}^{2}\right)\varphi_{m}\left(z\right) = \boldsymbol{0},
\end{equation}
where $\omega$ is the angle frequency of sound wave and $c\left(z\right)$ is the sound velocity. 

On the other hand, the displacement field includes the horizontal displacement field $u\left(r,z\right)$ and the vertical displacement field $w\left(r,z\right)$, which can be computed from the horizontal displacement modes $d^{(1)}_{m}\left(z\right)$ and the vertical displacement modes $d^{(2)}_{m}\left(z\right)$ \cite{aki2002quantitative},
\begin{equation*}
	u\left(r,z\right) = \sum_{m}\frac{d^{(1)}_{m}(z)}{8c_{m}U_{m}I_{m}}\sqrt{\frac{2}{\pi k_{m}r}}e^{\imath\left(k_{m}r-\frac{\pi}{4}\right)} \left\{\left.k_{m}d^{(1)}_{m}(z_{s})+\frac{\mathrm{d}d^{(2)}_{m}(z)}{\mathrm{d}z}\right|_{z_{s}}\right\},
\end{equation*}
\begin{equation*}
	w\left(r,z\right) = \sum_{m}\frac{d^{(2)}_{m}(z)}{8c_{m}U_{m}I_{m}}\sqrt{\frac{2}{\pi k_{m}r}}e^{\imath\left(k_{m}r+\frac{\pi}{4}\right)} \left\{\left.k_{m}d^{(1)}_{m}(z_{s})+\frac{\mathrm{d}d^{(2)}_{m}(z)}{\mathrm{d}z}\right|_{z_{s}}\right\},
\end{equation*}
where $U_{m}=\frac{\mathrm{d}\omega}{\mathrm{d}k_{m}}$ is the group velocity of the displacement mode, $c_{m}=\frac{\omega}{k_{m}}$ is the phase velocity of the displacement mode and $I_{m}=\frac{1}{2}\int \rho(z)\left({d^{(1)}_{m}}^{2}(z)+{d^{(2)}_{m}}^{2}(z)\right)\mathrm{d}z$. 
Denote $\xi_m\left(z\right) = \left(\xi_{m,1}\left(z\right), \xi_{m,2}\left(z\right), \xi_{m,3}\left(z\right), \xi_{m,4}\left(z\right)\right)^{\rm T} = \left(\frac{d^{(1)}_{m}\left(z\right)}{ik_m}, \imath d^{(2)}_{m}\left(z\right), \frac{\tau^{(zx)}_{m}\left(z\right)}{\imath k_m}, \tau^{(zz)}_{m}\left(z\right)\right)^{\rm T}$ be the stress-displacement mode, where $\tau^{(zx)}_{m}\left(z\right)$ and $\tau^{(zz)}_{m}\left(z\right)$ are tangential stress mode and normal stress mode, respectively. The stress-displacement mode satisfies the following system of differential equations,
\begin{equation}\label{equation: ice}
	\begin{cases}
		\frac{d\xi_{m,1}\left(z\right)}{dz} = -\xi_{m,2}\left(z\right) + \frac{1}{\mu(z)} \xi_{m,4}\left(z\right), \\
		\frac{d\xi_{m,2}\left(z\right)}{dz} = \frac{\lambda(z)k_m^2}{\lambda(z)+2\mu(z)} \xi_{m,1}\left(z\right) + \frac{1}{\lambda(z)+2\mu(z)} \xi_{m,4}\left(z\right), \\
		\frac{d\xi_{m,3}\left(z\right)}{dz} = \left(\frac{4\mu(z)\left(\lambda(z)+\mu(z)\right)k_m^2}{\lambda(z)+2\mu(z)} - \rho(z)\omega^2\right) \xi_{m,1}\left(z\right) - \frac{\lambda(z)}{\lambda(z)+2\mu(z)}\xi_{m,4}\left(z\right), \\
		\frac{d\xi_{m,4}\left(z\right)}{dz} = -\rho(z)\omega^2 \xi_{m,2}\left(z\right) + k_m^2 \xi_{m,3}\left(z\right),
	\end{cases}
\end{equation}
where $\lambda(z)$ and $\mu(z)$ are the Lam\'e coefficients of ice.

\begin{figure}[htbp]
	\centering
	\includegraphics[scale=0.6]{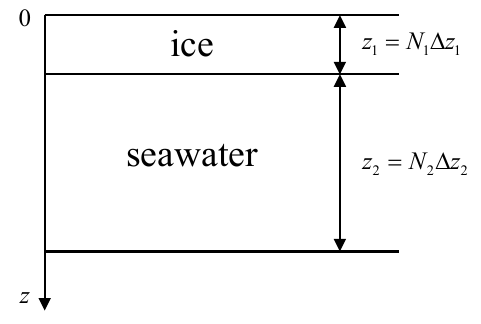}
	\caption{A simple ocean environment model.}
	\label{model}
\end{figure}

A simplified ocean environment model is depicted in Figure \ref{model}, the depths of ice and seawater are discretized into $z_{1}=N_{1}\Delta z_{1}$ and $z_{2}=N_{2}\Delta z_{2}$. Coupled with the boundary conditions and by finite difference methods, Equation \eqref{equation: seawater} and \eqref{equation: ice} can be transformed into the GEPs of the form
\begin{equation}
	AV=BV\Sigma,
    \label{appendix equation: GEPs}
\end{equation}
where $A$ and $B$ are sparse structured matrices, $V$ is a matrix whose columns are generalized eigenvectors which store modes $\varphi_{m}\left(z\right)$, $d^{(1)}_{m}\left(z\right)$, $d^{(2)}_{m}\left(z\right)$, and $\Sigma$ is a diagonal matrix with its diagonal elements being generalized eigenvalues which are the square of the horizontal wave numbers $k_{m}$. The dimensions of $A,B$ are $4N_{1}+N_{2}+5$, which are typically quite large. All non-zero-value elements are located near the diagonal with strong repeatability. We present the structure of matrices $A$ and $B$ in Figure~\ref{appendix matrix AB in GEPs}.
\begin{figure}[h!]
    \centering
    \begin{minipage}{0.48\linewidth}
        \centering
        \includegraphics[width=\linewidth]{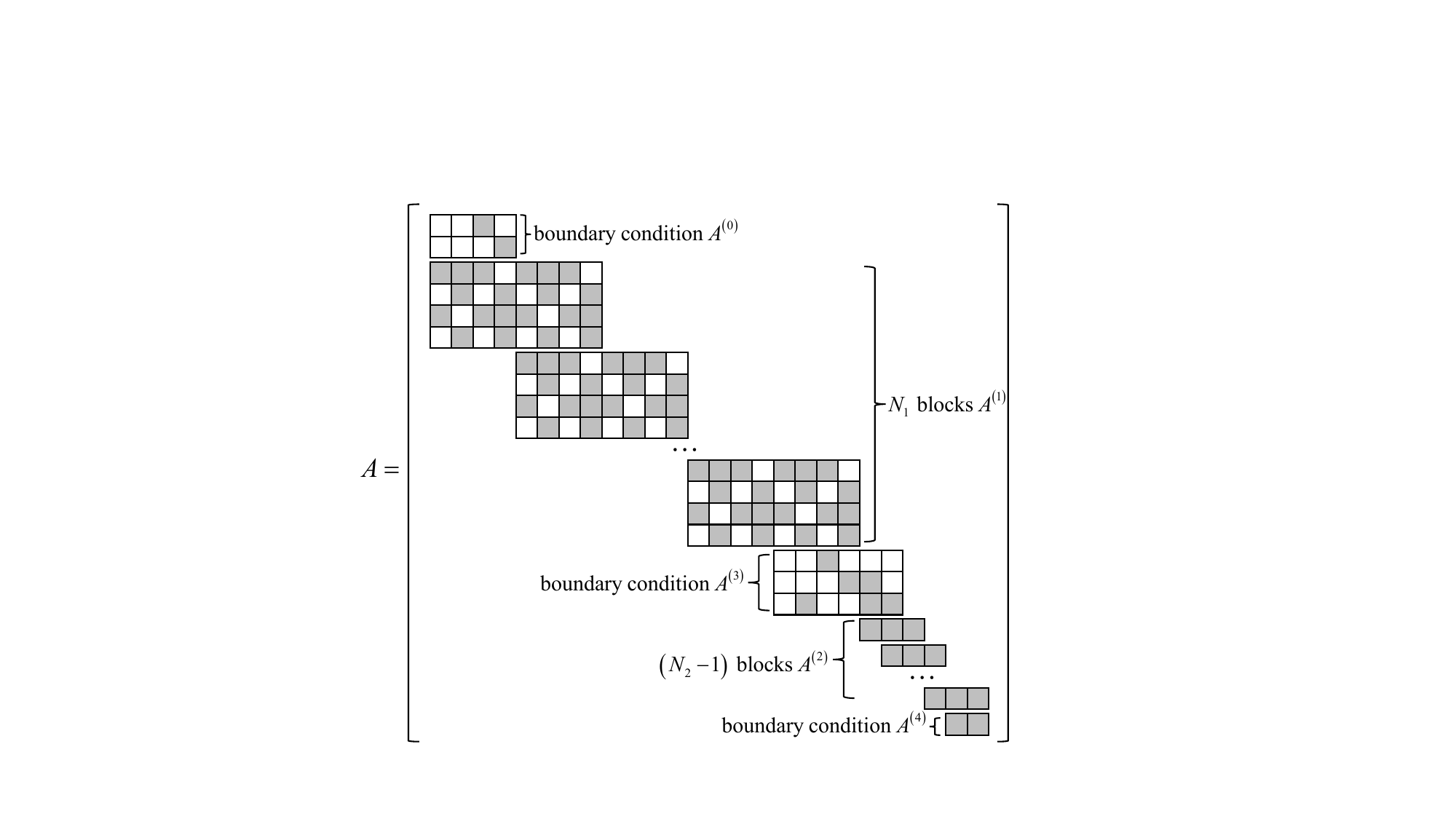}
    \end{minipage}
    \begin{minipage}{0.48\linewidth}
        \centering
        \includegraphics[width=\linewidth]{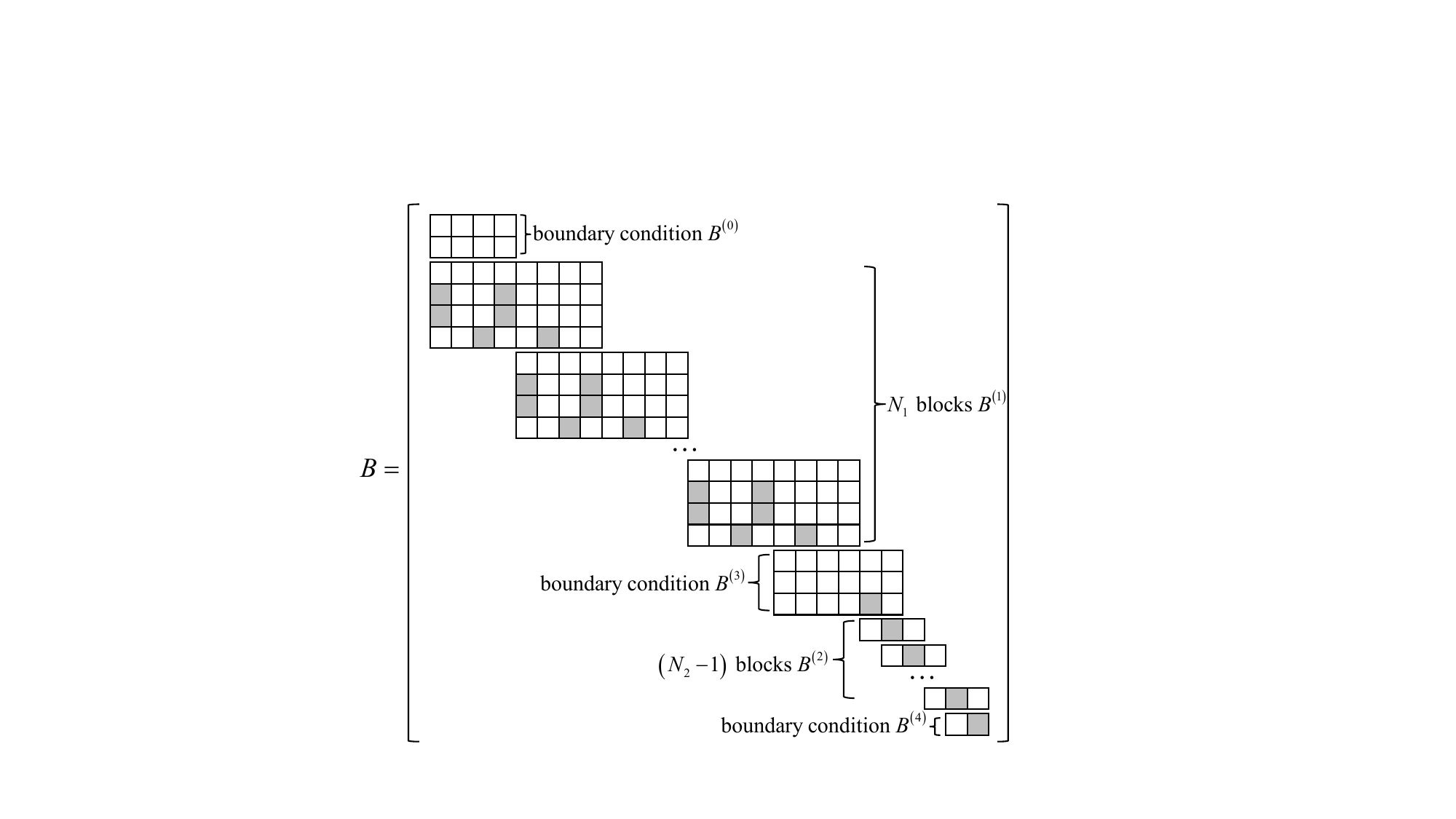}
    \end{minipage}
    \caption{The matrices $A$ and $B$ of GEPs in Equation~\eqref{appendix equation: GEPs}, where the gray and white squares represent a non-zero element and a zero element, respectively.}
    \label{appendix matrix AB in GEPs}
\end{figure}

The sub-matrix $\{A^{(i)}\}_{i=0}^{4}$ and  $\{B^{(i)}\}_{i=0}^{4}$ in the Figure~\ref{appendix matrix AB in GEPs} are defined as

\begin{equation*}
    A^{(0)} = \begin{pmatrix}
        0 & 0 & a_3 & 0 \\
        0 & 0 & 0 & a_3
    \end{pmatrix},\quad
    A^{(1)} = \begin{pmatrix}
        a_0 & a_2 & a_4 & 0 & a_7 & a_2 & a_4 & 0 \\
        0 & a_0 & 0 & a_5 & 0 & a_7 & 0 & a_5 \\
        a_{1} & 0 & a_{0} & a_{6} & a_1 & 0 & a_{7} & a_{6} \\
        0 & a_{1} & 0 & a_{0} & 0 & a_{1} & 0 & a_{7} \\
    \end{pmatrix},\quad
    A^{(2)} = \begin{pmatrix}
        a_3 & a_{10} & a_3
    \end{pmatrix},
\end{equation*}
\begin{equation*}
    A^{(3)} = \begin{pmatrix}
        0 & 0 & a_{3} & 0 & 0 & 0 \\	
	0 & 0 & 0 & a_{3} & a_{3} & 0 \\	
	0 & a_{8} & 0 & 0 & a_{9} & a_{3}
    \end{pmatrix},\quad
    A^{(4)} = \begin{pmatrix}
        a_{11} & a_{12}
    \end{pmatrix}, 
\end{equation*}
\begin{gather*}
    B^{(0)} = \begin{pmatrix}
        0 & 0 & 0 & 0 \\
        0 & 0 & 0 & 0
    \end{pmatrix},\quad
    B^{(1)} = \begin{pmatrix}
        0 & 0 & 0 & 0 & 0 & 0 & 0 & 0 \\
        b_{0} & 0 & 0 & b_{0} & 0 & 0 & 0 & 0 \\
        b_{2} & 0 & 0 & b_{2} & 0 & 0 & 0 & 0 \\
        0 & 0 & b_{1} & 0 & 0 & b_{1} & 0 & 0
    \end{pmatrix},\quad
    B^{(2)} = \begin{pmatrix}
        0 & b_4 & 0
    \end{pmatrix},\\
    B^{(3)} = \begin{pmatrix}
        0 & 0 & 0 & 0 & 0 & 0 \\
	0 & 0 & 0 & 0 & 0 & 0 \\
	  0 & 0 & 0 & 0 & b_{3} & 0
    \end{pmatrix},\quad
    B^{(4)} = \begin{pmatrix}
        0 & b_5
    \end{pmatrix}.
\end{gather*}

\begin{table}[h!]
    \centering
    \begin{tabular}{|c|l|}
        \hline
        \textbf{Keys} & \multicolumn{1}{c|}{\textbf{Values}} \\
        \hline
        $0$ & $\left\{(a_{ij},i,j): a_{ij} = a_0, (i,j)\in\left(j+2, [0,4N_1-1]\right)\right\}$ \\
        \hline
        $1$ & $\left\{(a_{ij},i,j): a_{ij} = a_1, (i,j)\in\left(j+4, \left\{[0,4N_1-3]:{\rm mod}(j,4)=0\mbox{ or }1\right\}\right)\right\}$ \\
        \hline
        $2$ & $\left\{(a_{ij},i,j): a_{ij} = a_1, (i,j)\in\left(j, \left\{[4,4N_1+1]:{\rm mod}(j,4)=0\mbox{ or }1\right\}\right)\right\}$ \\
        \hline
        $3$ & $\left\{(a_{ij},i,j): a_{ij} = a_2, (i,j)\in\left(j+1, \left\{[1,4N_1-3]:{\rm mod}(j,4)=1\right\}\right)\right\}$ \\
        \hline
        $4$ & $\left\{(a_{ij},i,j): a_{ij} = a_2, (i,j)\in\left(j-3, \left\{[5,4N_1+1]:{\rm mod}(j,4)=1\right\}\right)\right\}$ \\
        \hline
        $5$ & \makecell{$\left\{(a_{ij},i,j): a_{ij} = a_3, (i,j)\in\left(j+1, \left\{[4N_1+4,4N_1+N_2+2]\right\}\right)\right\}$\\
        $\bigcup \left\{(a_{ij},i,j): a_{ij} = a_3, (i,j)\in\left(j, \left\{4N_1+2,4N_1+3\right\}\right)\right\}$\\
        $\bigcup\left\{(a_{ij},i,j): a_{ij} = a_3, (i,j)\in\left(j-2, \left\{2,3\right\}\right)\right\}$} \\
        \hline
        $6$ & $\left\{(a_{ij},i,j): a_{ij} = a_3, (i,j)\in\left(j-1, \left\{[4N_1+4,4N_1+N_2+4]\right\}\right)\right\}$ \\
        \hline
        $7$ & $\left\{(a_{ij},i,j): a_{ij} = a_4 (i,j)\in\left(j, \left\{[2,4N_1-2]:{\rm mod}(j,4)=2\right\}\right)\right\}$ \\
        \hline
        $8$ & $\left\{(a_{ij},i,j): a_{ij} = a_4, (i,j)\in\left(j-4, \left\{[6,4N_1+2]:{\rm mod}(j,4)=2\right\}\right)\right\}$ \\
        \hline
        $9$ & $\left\{(a_{ij},i,j): a_{ij} = a_5, (i,j)\in\left(j-4,\left\{[7,4N_1+3]:{\rm mod}(j,4)=3\right\}\right)\right\}$ \\
        \hline
        $10$ & $\left\{(a_{ij},i,j): a_{ij} = a_5, (i,j)\in\left(j, \left\{[3,4N_1-1]:{\rm mod}(j,4)=3\right\}\right)\right\}$ \\
        \hline
        $11$ & $\left\{(a_{ij},i,j): a_{ij} = a_6, (i,j)\in\left(j-3, \left\{[7,4N_1+3]:{\rm mod}(j,4)=3\right\}\right)\right\}$ \\
        \hline
        $12$ & $\left\{(a_{ij},i,j): a_{ij} = a_6, (i,j)\in\left(j+1, \left\{[3,4N_1-1]:{\rm mod}(j,4)=3\right\}\right)\right\}$ \\
        \hline
        $13$ & $\left\{(a_{ij},i,j): a_{ij} = a_7, (i,j)\in\left(j-2, [4,4N_1+3]\right)\right\}$ \\
        \hline
        $14$ & $\left\{(a_{ij},i,j): a_{ij} = a_8, (i,j)\in\left(j+3, \left\{4N_1+1\right\}\right)\right\}$ \\
        \hline
        $15$ & $\left\{(a_{ij},i,j): a_{ij} = a_{9}, (i,j)\in\left(j, \left\{4N_1+4\right\}\right)\right\}$ \\
        \hline
        $16$ & $\left\{(a_{ij},i,j): a_{ij} = a_{10}, (i,j)\in\left(j, [4N_1+5,4N_1+N_2+3]\right)\right\}$ \\
        \hline
        $17$ & $\left\{(a_{ij},i,j): a_{ij} = a_{11}, (i,j)\in\left(j+1, \left\{4N_1+N_2+3\right\}\right)\right\}$ \\
        \hline
        $18$ & $\left\{(a_{ij},i,j): a_{ij} = a_{12}, (i,j)\in\left(j, \left\{4N_1+N_2+4\right\}\right)\right\}$ \\
        \hline
    \end{tabular}
    \caption{The dictionary data structure of the matrix $A$ of GEPs in Equation~\eqref{equation: GEPs}.}
    \label{table: dictionary of A}
\end{table}

\begin{table}[h!]
    \centering
    \begin{tabular}{|c|l|}
        \hline
        \textbf{Keys} & \multicolumn{1}{c|}{\textbf{Values}} \\
        \hline
        $0$ & $\left\{(b_{ij},i,j): b_{ij} = b_0, (i,j)\in\left(j+3, \left\{[0,4N_1-4]:{\rm mod}(j,4)=0\right\}\right)\right\}$ \\
        \hline
        $1$ & $ \left\{(b_{ij},i,j): b_{ij} = b_0, (i,j)\in\left(j-1, \left\{[4,4N_1]:{\rm mod}(j,4)=0\right\}\right)\right\}$ \\
        \hline
        $2$ & $\left\{(b_{ij},i,j): b_{ij} = b_1, (i,j)\in\left(j+3, \left\{[2,4N_1-2]:{\rm mod}(j,4)=2\right\}\right)\right\}$ \\
        \hline
        $3$ & $\left\{(b_{ij},i,j): b_{ij} = b_1, (i,j)\in\left(j-1, \left\{[6,4N_1+2]:{\rm mod}(j,4)=2\right\}\right)\right\}$ \\
        \hline
        $4$ & $\left\{(b_{ij},i,j): b_{ij} = b_2, (i,j)\in\left(j+4, \left\{[0,4N_1-4]:{\rm mod}(j,4)=0\right\}\right)\right\}$ \\
        \hline
        $5$ & $\left\{(b_{ij},i,j): b_{ij} = b_2, (i,j)\in\left(j, \left\{[4,4N_1]:{\rm mod}(j,4)=0\right\}\right)\right\}$ \\
        \hline
        $6$ & $\left\{(b_{ij},i,j): b_{ij} = b_3, (i,j)\in\left(j, \left\{4N_1+4\right\}\right)\right\}$ \\
        \hline
        $7$ & $\left\{(b_{ij},i,j): b_{ij} = b_4, (i,j)\in\left(j, [4N_1+5,4N_1+N_2+3]\right)\right\}$ \\
        \hline
        $8$ & $\left\{(b_{ij},i,j): b_{ij} = b_5, (i,j)\in\left(j, \left\{4N_1+N_2+4\right\}\right)\right\}$ \\
        \hline
    \end{tabular}
    \caption{The dictionary data structure of the matrix $B$ of GEPs in Equation~\eqref{equation: GEPs}.}
    \label{table: dictionary of B}
\end{table}

\end{document}